\documentclass[prx,aps,superscriptaddress,footinbib,twocolumn]{revtex4-2}

\usepackage{amsthm,amsmath,amssymb}
\usepackage{mathrsfs}

\newcommand{\ket}[1]{\vert{ #1 }\rangle}
\newcommand{\bra}[1]{\langle{ #1 }\vert}

\newcommand{\braket}[2]{\langle #1 \vert #2 \rangle}

\newcommand{\E}{\operatorname{E}}

\usepackage{amsthm}
\newtheorem{theorem}{Theorem}
\newtheorem{proposition}{Proposition}
\newtheorem{lemma}{Lemma}

\theoremstyle{definition}
\newtheorem{definition}{Definition}
\theoremstyle{remark}

\usepackage{amssymb}
\usepackage{booktabs}

\usepackage{multirow}
\usepackage{graphicx}
\usepackage{epstopdf}
\newcommand{\figures}{.}

\usepackage{algorithm}
\usepackage[noend]{algpseudocode}

\usepackage{natbib}

\usepackage[colorlinks=true,citecolor=blue,linkcolor=magenta]{hyperref}

\usepackage{xcolor}

\usepackage{comment}

\begin{document}

\title{Time-Efficient Logical Operations on Quantum Low-Density Parity Checks Codes}

\author{Guo Zhang}
\affiliation{Graduate School of China Academy of Engineering Physics, Beijing 100193, China}

\author{Ying Li}
\email{yli@gscaep.ac.cn}
\affiliation{Graduate School of China Academy of Engineering Physics, Beijing 100193, China}

\begin{abstract}
We propose schemes capable of measuring an arbitrary set of commutative logical Pauli operators in time independent of the number of operators. The only condition is commutativity, a fundamental requirement for simultaneous measurements in quantum mechanics. Quantum low-density parity check (qLDPC) codes show great promise for realizing fault-tolerant quantum computing. They are particularly significant for early fault-tolerant technologies as they can encode many logical qubits using relatively few physical qubits. By achieving simultaneous measurements of logical operators, our approaches enable fully parallelized quantum computing, thus minimizing computation time. Our schemes are applicable to any qLDPC codes and maintain the low density of parity checks while measuring multiple logical operators simultaneously. These results enhance the feasibility of applying early fault-tolerant technologies to practical problems. 
\end{abstract}

\maketitle

\section{Introduction}

Quantum error correction is crucial for many quantum computing applications, such as breaking cryptographic systems and simulating quantum many-body physics \cite{OGorman2017,Babbush2018}. The primary challenge of quantum error correction lies in the substantial number of physical qubits required for encoding \cite{Fowler2012}. Quantum low-density parity check (qLDPC) codes offer an advantage in this regard because of their low overhead \cite{Gottesman2014,Breuckmann2021}. Recent progresses demonstrate that the long-range connectivity needed to implement low-overhead qLDPC codes is feasible in neutral atom and ion trap systems \cite{Bluvstein2023,Evered2023,DeCross2024}. Furthermore, numerical results indicate that qLDPC codes can tolerate relatively high physical error rates \cite{grospellier2021,Xu2024,Bravyi2024}. These advancements underscore the potential of qLDPC codes as a pivotal pathway to achieving fault-tolerant quantum computing \cite{Nielsen2003}. 

A promising method for implementing logical operations on qLDPC codes is lattice surgery \cite{Horsman2012,Vuillot2019,Litinski2019}. In this approach, an ancilla system is coupled with the memory enabling the measurement of logical qubits \cite{Cohen2022,Xu2024}. However, multiple logical measurements involving the same physical qubits cannot be executed simultaneously to maintain the low density of parity checks. This issue hinders the parallelization of logical operations and can potentially increase the time required for quantum computations. Since fundamental physical operations on qubits are considerably time-consuming, the complexity resulting from the lack of parallelization is particularly important. It may ultimately limit the practical applications of quantum computing. 

In this paper, we propose two schemes for simultaneous measurements on multiple logical operators. Our schemes apply to general qLDPC codes, including subsystem codes \cite{Poulin2005}. We achieve {\it ultimate parallelism} allowing for the measurement of an arbitrary set of commutative logical Pauli operators in time independent of the number of operators. For instance, the logical operator set could be $\{\bar{X}_1,\bar{X}_2\bar{Z}_3,\bar{Y}_2\bar{Y}_3\bar{X_4}\cdots\bar{Z}_k,\ldots\}$, in which logical operators may even overlap with each other on the same logical qubits. Here, commutativity is the only condition, which is a fundamental requirement in quantum mechanics: simultaneous measurements are permitted if and only if the operators commute with each other \cite{L.D.Landau1981}. The simultaneous measurements, supplemented with the preparation of magic states, enable fully parallelized universal quantum computing. 

\begin{figure}[tbp]
\centering
\includegraphics[width=\linewidth]{\figures/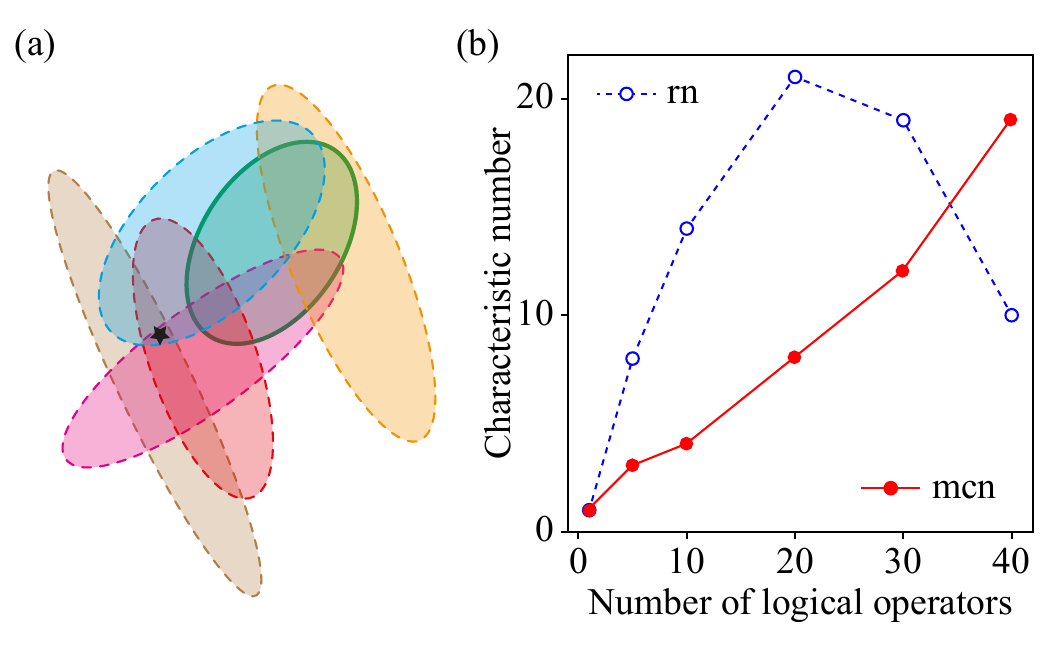}
\caption{
(a) Logical operators and their supports. Each circle represents the support of a logical operator. In the region marked by the star, each physical qubit is in supports of four logical operators. The logical operator marked by the green solid-line circle is contained in supports of other logical operators. (b) Median values of the maximum crowd number (mcn) and redundancy number (rn) for a quantum low-density parity check code. The set $\Sigma$ consists of $Z$ logical operators acting non-trivially on up to $L = 5$ logical qubits. 
}
\label{fig:overlap}
\end{figure}

The simultaneous measurements are achieved through two types of ancilla systems: measurement stickers and branch stickers. The function of a sticker is determined by a linear code, referred to as the glue code. One of the schemes, termed devised sticking, employs a single measurement sticker. By adjusting the glue code, we can realize the desired simultaneous measurement. In the other scheme, termed brute-force branching, we concatenate branch and measurement stickers to propagate the logical operators to different stickers for simultaneous measurement. Both schemes maintain the low density of parity checks. 

\section{Problems}

The difficulty in simultaneously measuring an arbitrary set of logical operators arises from the overlap of logical operators; see Fig. \ref{fig:overlap}(a). In lattice surgery, the method to measure a logical operator involves coupling an ancilla system to the physical qubits within the support of this operator. To measure multiple logical operators simultaneously, we can use multiple ancilla systems, with each ancilla system measuring one logical operator by coupling it to the corresponding support, as proposed in Ref. \cite{Cohen2022}. However, because of the overlap of logical operators, this approach might result in some physical qubits being coupled to multiple ancilla systems, thereby violating the LDPC condition. This problem has been noticed in Refs. \cite{Cohen2022,Xu2024}. Another method for simultaneously measuring multiple logical operators is to couple an ancilla system to the union of supports of all the logical operators to be measured. Because of the overlap of logical operators, the union of the supports might contain logical operators that do not need to be measured. However, these redundant logical operators may also be measured by the ancilla system, leading to incorrect logical operations. In this work, we address both issues, enabling the simultaneous measurement of an arbitrary set of logical Pauli operators. 

Because of their high encoding rate, qLDPC codes are prone to logical-operator overlap. In Appendix \ref{app:comparison}, we justify this argument through a rigorous analysis showing that the problem of overlap cannot be solved through optimising the representatives of logical operators. We can use two quantities to characterize the overlap, corresponding to the two issues mentioned above, respectively. Let $\Sigma = \{\sigma_1,\sigma_2,\ldots,\sigma_q\}$ be a subset of $Z$ ($X$) logical operators. For a physical qubit, its crowd number is the number of operators in $\Sigma$ acting non-trivially on that physical qubit. The redundancy number of $\Sigma$ is the number of $Z$ ($X$) logical operators that are contained in the union of supports but not in $\Sigma$ (we only count independent operators). See Appendix \ref{app:preliminaries} for formal definitions. In Fig. \ref{fig:overlap}(b), we demonstrate how these two quantities change with the size of $\Sigma$ using a [[1922,50,16]] code as an example \cite{Panteleev2021,Kovalev2012}. We find that the problem of logical-operator overlap becomes more severe as the size of $\Sigma$ grows. 

\begin{figure}[tbp]
\centering
\includegraphics[width=\linewidth]{\figures/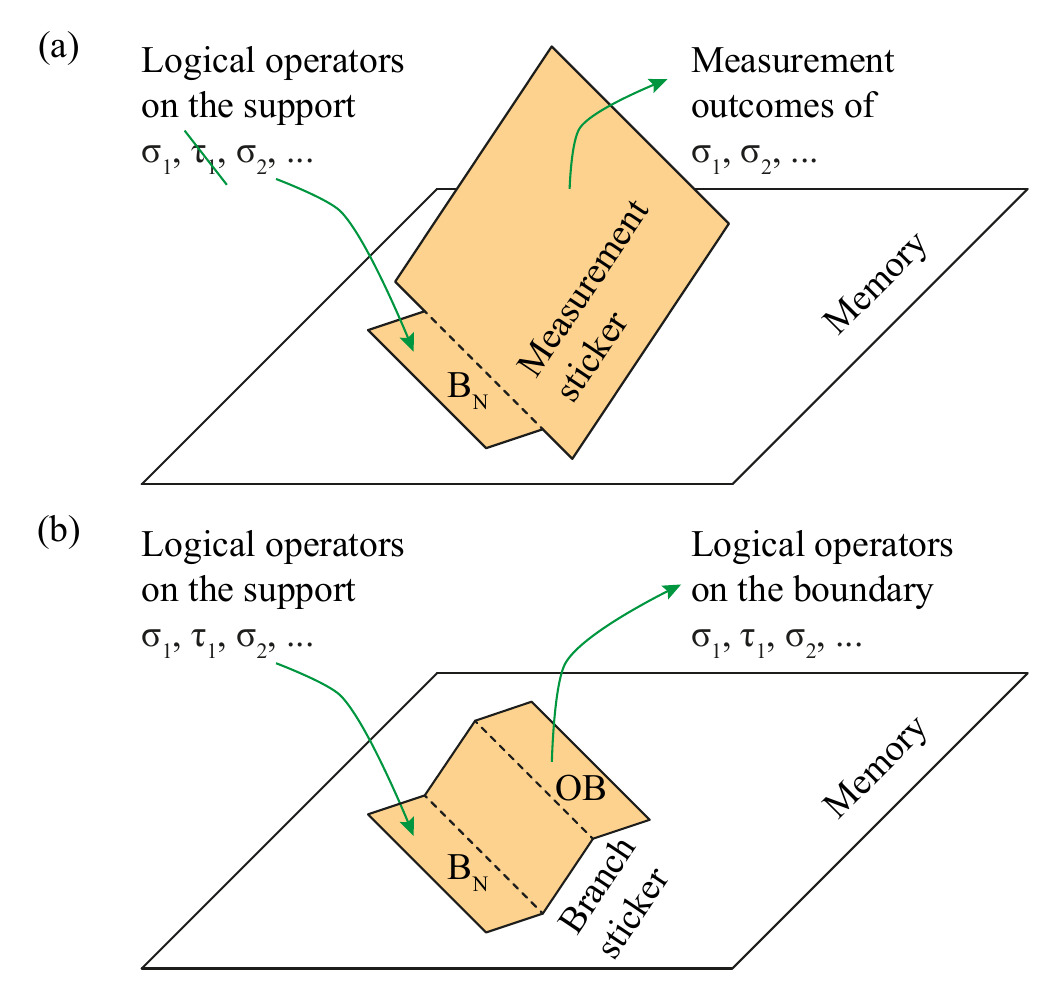}
\caption{
(a) A measurement sticker pasted on $\mathcal{B}_N$ on the memory. Logical operators contained in $\mathcal{B}_N$ include $\sigma_1,\sigma_2,\ldots$ and $\tau_1,\ldots$. We can choose to measure $\sigma_1,\sigma_2,\ldots$ by designing an appropriate measurement sticker. (b) A branch sticker pasted on $\mathcal{B}_N$ on the memory. Logical operators contained in $\mathcal{B}_N$ are transferred to the open boundary (OB) of the branch sticker. 
}
\label{fig:schemes}
\end{figure}

\section{Schemes}

We propose two methods to measure multiple logical operators simultaneously, devised sticking and brute-force branching. These two methods solve the two problems caused by logical-operator overlap, respectively. In devised sticking, we use only one ancilla system, called measurement sticker, and couple it with a subset $\mathcal{B}_N$ of physical qubits on the memory. Here, $\mathcal{B}_N$ is the union of supports of the logical operators to be measured. By designing an appropriate measurement sticker, we can measure any selected subset (rather than all) of logical operators contained in $\mathcal{B}_N$, as shown in Fig. \ref{fig:schemes}(a). In brute-force branching, we use another type of ancilla system called branch sticker. Unlike a measurement sticker, the role of a branch sticker is to transfer logical operators from the memory to the sticker (specifically to a subset of physical qubits on the sticker, called its open boundary), as shown in Fig. \ref{fig:schemes}(b). Through the concatenation of branch stickers, we can transfer the logical operators to different stickers for measurement, thereby eliminating the overlap between logical operators. Then, we can measure each logical qubit using a measurement sticker without violating the LDPC condition. 

\begin{figure}[tbp]
\centering
\includegraphics[width=\linewidth]{\figures/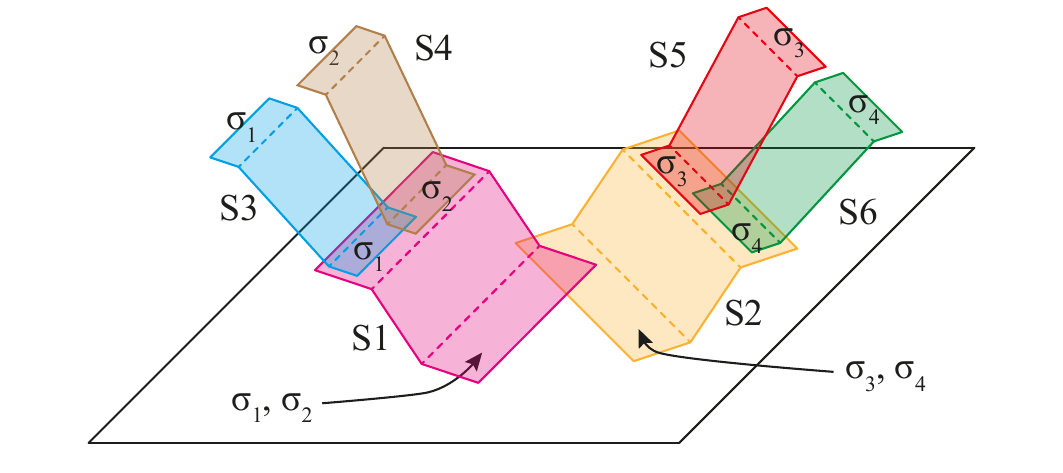}
\caption{
Brute-force branching for separating four logical operators. 
}
\label{fig:branching}
\end{figure}

In brute-force branching, we separate $q$ independent logical operators by concatenating branch stickers for $O(\log_2 q)$ levels. Fig. \ref{fig:branching} illustrates how to separate four overlapping logical operators. First, we paste the level-1 branch sticker $S1$ ($S2$) on the supports of $\sigma_1$ and $\sigma_2$ ($\sigma_3$ and $\sigma_4$), transferring $\sigma_1$ and $\sigma_2$ ($\sigma_3$ and $\sigma_4$) to the open boundary of $S1$ ($S2$). Then, we paste the level-2 branch sticker $S3$, $S4$, $S5$ and $S6$ on open boundaries of $S1$ and $S2$, transferring $\sigma_1$, $\sigma_2$, $\sigma_3$ and $\sigma_4$ to open boundaries of level-2 branch stickers, respectively. In this way, we can separate the $q$ logical operators within $O(\log_2 q)$ levels, i.e. we paste two branch stickers on each lower-level branch sticker. Once separated, we can measure each logical operator by pasting a measurement sticker on the open boundary of the corresponding highest-level branch sticker ($S3$, $S4$, $S5$ and $S6$ in Fig. \ref{fig:branching}). 

Using either of the two methods, devised sticking and brute-force branching, we can simultaneously measure an arbitrary set of $X$ or $Z$ logical operators. In what follows, we focus on the simultaneous measurement of $Z$ logical operators. The measurement of $X$ logical operators is similar. For general logical Pauli operators, we can measure them through the $X$ and $Z$ measurements. We give the protocol for the simultaneous measurement of general logical Pauli operators in Appendix \ref{app:general}. 

\begin{figure}[tbp]
\centering
\includegraphics[width=\linewidth]{\figures/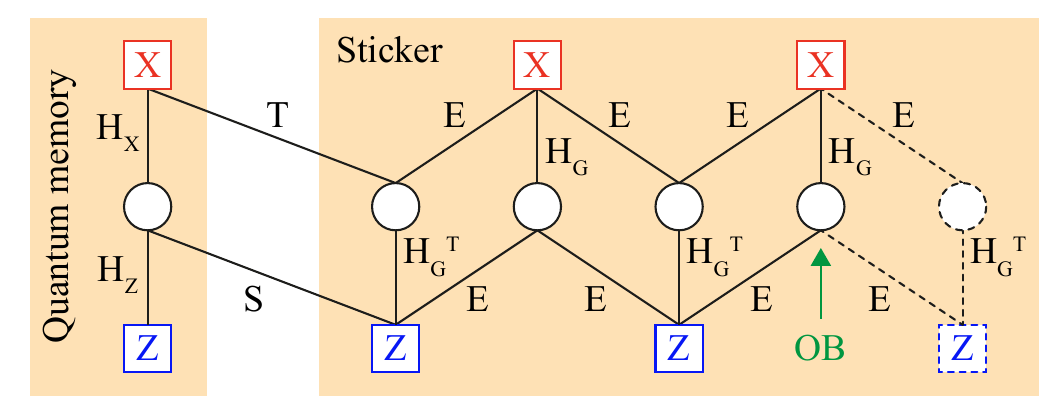}
\caption{
Tanner graph of a measurement sticker pasted on the quantum memory. Each circle represents a set of qubits, each square represents a set of $X$ or $Z$ parity checks, and each edge represents a check matrix. $H_X$ and $H_Z$ are check matrices of the memory. $H_G$ is the check matrix of the glue code. $S$ and $T$ are pasting matrices; the support of the row space of $S$ is the qubit subset $\mathcal{B}_N$. $E$ is the identity matrix. If removing the dashed circle, square and edges, we obtain a branch sticker, and the arrow indicates the open boundary (OB). The length of a sticker $d_R$ is the number of $Z$ squares. See Appendices \ref{app:stickers} and \ref{app:deformed_codes} for the matrix representation of the corresponding code. 
}
\label{fig:sticker}
\end{figure}

\section{Stickers and glue codes}

A sticker is a hypergraph product (HGP) code \cite{Tillich2014,Kovalev2012}, as shown in Fig. \ref{fig:sticker}. For a measurement sticker, one of the two linear codes that generate the HGP code is called the glue code; the other linear code is a repetition code. Branch stickers are similar. By deleting a bit in the repetition code, we obtain the HGP code for a branch sticker. Stickers are coupled to the memory through two matrices, $S$ and $T$, called pasting matrices. Let $H_G$ and $H_X$ be the check matrices for the glue code and $X$ operators of the memory, respectively. We say that the glue code is {\it compatible} with the memory if and only if there exist matrices $S$ and $T$ that satisfy the equation $H_XS^\mathrm{T} = TH_G$. 

The glue code determines which logical operators the sticker acts on. For a compatible glue code, $(\mathrm{ker}H_G)S \subseteq \mathrm{ker}H_X$, i.e. codewords of the glue code correspond to the stabilizer, logical and gauge operators of the memory. The sticker acts on such operators. Therefore, we can realize the desired logical measurement by designing the glue code. 

We use two types of glue codes. Let $\Sigma$ be the set of $Z$ logical operators to be acted on. For a measurement sticker, we need to choose an appropriate glue code such that only operators in $\Sigma$ are measured, and no other logical operators are measured. We refer to such a glue code as {\it finely devised} for $\Sigma$. For a branch sticker, the requirement for the glue code is weaker. A branch sticker does not measure any logical operators (and thus does not destroy any logical information), so we do not need to exclude logical operators outside $\Sigma$. That is, we need a glue code that can transfer the operators in $\Sigma$, but it may also transfer other logical operators simultaneously. We refer to such a glue code as {\it coarsely designed} for $\Sigma$. We provide rigorous definitions of the two types of glue codes in Appendix \ref{app:glue_code}. 

\begin{theorem}
For an arbitrary qLDPC code and an arbitrary set of $Z$ logical operators $\Sigma$, there exist coarsely and finely devised glue codes for $\Sigma$. The check matrix of the glue code $H_G$ and the corresponding pasting matrices $S$ and $T$ have a weight upper bounded by a factor independent of code parameters, i.e. satisfy the LDPC condition. Let $r_G\times n_G$ be the dimension of $H_G$. For the coarsely devised glue code, $n_G,r_G = O(n_N)$, where $n_N$ denotes the size of the union of supports for $\Sigma$. For the finely devised glue code, $n_G,r_G = O(n_N+(k_N-q)q)$, where $q$ is the number of independent operators in $\Sigma$, and $k_N$ is the number of independent $Z$ logical operators contained in the union of supports. 
\label{the:sticking}
\end{theorem}

Notice that $k_N-q$ is the redundancy number of $\Sigma$. In Appendix \ref{app:sticking}, we provide a more formal statement of the above theorem and its proof. The proof contains algorithms for generating coarsely and finely devised glue codes. 

\section{Deformed codes}

By pasting one or more stickers to the memory, we obtain a deformed code. For a single sticker, the deformed code is shown in Fig. \ref{fig:sticker}. For multiple stickers, we can construct the deformed code as follows: first, paste one sticker to the memory and treat the resulting deformed code as the new memory; then, paste the second sticker to the new memory; and so on. By coupling all the stickers to the memory, we generate the final deformed code. Based on the generated deformed code, we can perform logical measurements using lattice surgery. 

The steps for lattice surgery are as follows: 1) Initialise the physical qubits on all stickers to the state $\ket{+}$; 2) Perform parity-check measurements according to the deformed code and repeat this for $d_T$ times; 3) Measure physical qubits on all stickers in the $X$ basis. 

According to Theorem \ref{the:sticking}, the deformed code is always a qLDPC code. Besides the LDPC condition, the deformed code also needs to have a sufficiently large code distance. The properties of the deformed code are given by the following theorem. 

\begin{theorem}
The deformed code is suitable for lattice surgery and can achieve the corresponding operations on $Z$ logical operators. Let $d$ be the code distance of the memory, and let $d_R$ be the code distance of the repetition code generating the sticker. For a measurement sticker, the code distance of the deformed code has a lower bound of $\min\{d/\vert S\vert,d_R\}$. For a branch sticker, the code distance of the deformed code has a lower bound of $d/\vert S\vert$. Here, $\vert S\vert$ is the norm of the matrix $S$ induced by the Hamming weight. 
\label{the:GLS}
\end{theorem}

Notice that we can always choose $S$ such that $\vert S\vert = 1$. In Appendix \ref{app:lattice_surgery}, we provide a more formal statement of the above theorem and its proof. Additionally, in Appendix \ref{app:comparison}, we present a detailed comparison of stickers used in our methods with ancilla systems proposed in Ref. \cite{Cohen2022}. 

\begin{figure}[tbp]
\centering
\includegraphics[width=\linewidth]{\figures/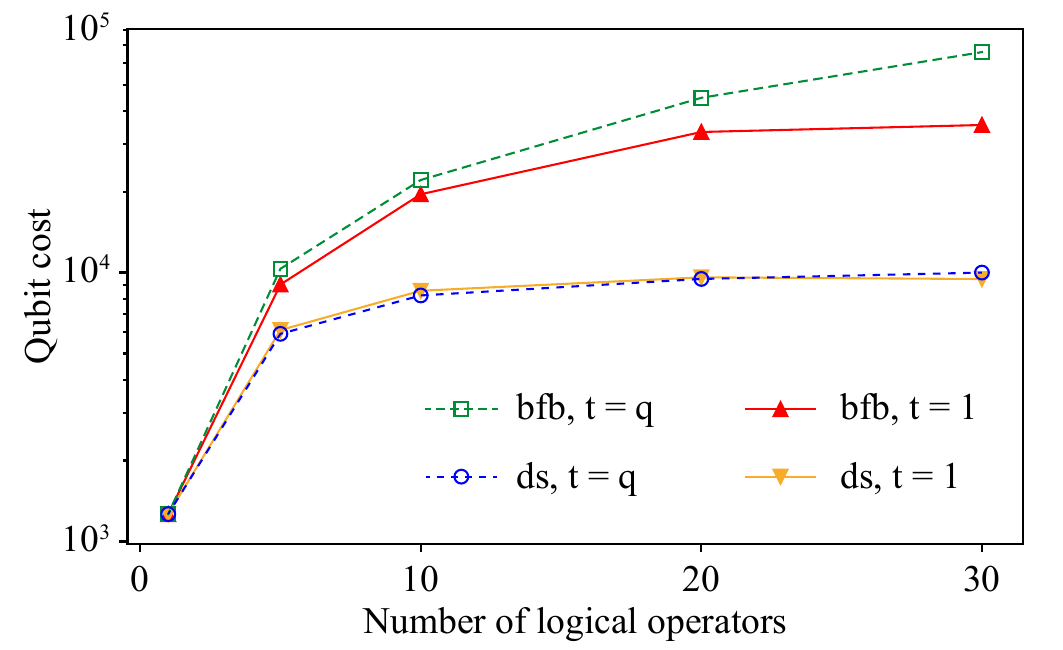}
\caption{
Median values of the qubit number required in simultaneous measurements. For each number of logical operators $q$, we randomly generate the operator set $\Sigma$ for one hundred times. Each set $\Sigma$ consists of $Z$ logical operators acting non-trivially on up to $L = 5$ logical qubits, and its logical thickness is $t$. For each $\Sigma$, we evaluate the qubit costs in devised sticking (ds) and brute-force branching (bfb). 
}
\label{fig:cost}
\end{figure}

\section{Costs}

In this work, we achieve ultimate parallelism. In conventional parallelism, we can operate logical qubits in parallel, i.e. we can apply operations simultaneously if they act on different logical qubits. This is the parallelism normally considered when compiling quantum circuits. On the surface code and HGP codes, protocols have been proposed for logical measurements in the type of conventional parallelism \cite{Litinski2019, Xu202407}. In ultimate parallelism, we reach the physical limit, i.e. commutativity is the only condition. We can apply measurements simultaneously as long as they commute, and we can apply multiple measurements simultaneously on the same logical qubit. Ultimate parallelism is more powerful than conventional parallelism. In Appendix \ref{app:parallelism}, we illustrate the difference with an example. 

The time required for simultaneous measurements depends on the parameter $d_T$ in lattice surgery. To suppress measurement errors, parity-check measurements need to be repeated for sufficiently many times, meaning $d_T$ must be sufficiently large. Suppose the memory code has parameters $[[n,k,d]]$. Usually, we take $d_T = \Theta(d)$ to balance measurement errors and data-qubit errors. Potentially, we could correct measurement errors and reduce the time cost to a constant by employing three-dimensional homological product codes \cite{Bravyi2014}, which are single-shot codes \cite{Bombin2015,Quintavalle2021}, to construct stickers or utilising code-inspired robust projective measurements \cite{Ouyang2024}. Regarding the number of logical operators to be measured, the time cost is independent of the operator number $q$. 

The qubit overhead depends on the measurement protocol. In Fig.~\ref{fig:cost}, we use a [[1922,50,16]] code as an example to illustrate the qubit costs in devised sticking and brute-force branching \cite{Panteleev2021,Kovalev2012}; the result of a [[578,162,5]] code is similar (see Appendix \ref{app:costs}). We find that the qubit cost in devised sticking is smaller than brute-force branching. In devised sticking, the qubit cost is $O(n_Ndq)$, where $n_N$ is the number of physical qubits in the support of the operator set $\Sigma$. See Appendix \ref{app:costs} for the bound analysis of both protocols. Although taking a finite $q$ is possible in practice, we may need to measure $q = \Theta(k)$ logical operators simultaneously to achieve ultimate parallelism. In this case, the bound $n_N = \Theta(n)$ is applicable. 

\begin{table}[!ht]
\begin{tabular}{ccccc}
\toprule
Code & Space $\times$ Time cost & Refs. \\
\midrule
Surface code & $\Theta(kd^2)\times\Theta(d) = \Theta(k^{5/2})$ & \cite{Litinski2019} \\
\midrule
HGP codes & $\Theta(k)\times O(k^{3/4}d) = O(k^{9/4})$ & \cite{Xu202407} \\
\midrule
qLDPC codes & $O(kd)\times\Theta(d) = O(k^2)$ & This work \\
\bottomrule
\end{tabular}
\caption{
The space (qubit) and time costs of implementing a layer of $\Theta(k)$ Clifford gates. We suppose that the gates are controlled-NOT, Hadamard and $S$ gates, and they are disjoint on logical qubits. The protocols for the surface code, HGP codes and qLDPC codes are lattice surgery \cite{Litinski2019}, GPPM \cite{Xu202407} and devised sticking, respectively. To estimate the space cost of devised sticking, we assume a family of qLDPC codes satisfying $n = \Theta(k)$. To estimate the spacetime cost, we assume $d = \Theta(k^{1/2})$ according to HGP codes. 
}
\label{table}
\end{table}

If reducing the goal to conventional parallelism, we can reduce the qubit cost in the bound analysis. We say two logical operators have a logical overlap if they act non-trivially on the same {\it logical qubit}. We say the operator set $\Sigma$ has a logical thickness of $t$ if there is a partition $\Sigma = \Sigma_1\cup\Sigma_2\cup\cdots$ such that $\vert \Sigma_l\vert\leq t$ for all subsets, and only operators in the same subset have logical overlaps. To measure such an operator set, the qubit cost is $O(n_Ndt)$. For conventional parallelism, we have $t = 1$. Although different in the bound analysis, numerical results illustrate that the qubit costs of conventional parallelism and ultimate parallelism are comparable; see Fig. \ref{fig:cost}. 

In Table \ref{table}, we compare devised sticking with protocols for the surface code and HGP codes. Our protocol is applicable to general qLDPC codes and has a smaller spacetime cost, i.e. the cost is reduced from $O(k^{9/4})$ to $O(k^2)$ for HGP codes. 

\section{Conclusions}

In this work, we propose two schemes of constructing deformed codes for lattice surgery, enabling the simultaneous measurements of arbitrary logical Pauli operators. We rigorously analyze the code distance and the weight of check matrices. We also estimate the number of qubits required for simultaneous measurements. Our schemes are flexible in the trade-off between qubit and time costs by choosing the operator number in each simultaneous measurement. When the time cost is minimized, we show that our protocol reduces the spacetime cost in logical operations on HGP codes; and even in this case, the qubit cost is smaller than the surface code. The results demonstrate that fully parallelized fault-tolerant quantum computing can be achieved on arbitrary qLDPC codes. 

\section{Acknowledgments}

This work is supported by the National Natural Science Foundation of China (Grant Nos. 12225507, 12088101) and NSAF (Grant No. U1930403). 

\section{Date availability}

The source codes for the numerical simulation are available at \cite{code}. 

\appendix

\begin{widetext}

\section{Preliminaries}
\label{app:preliminaries}

{\bf Subsystem codes.} We denote a CSS subsystem code \cite{Poulin2005} with a six-tuple $(H_X,H_Z,J_X,J_Z,F_X,F_Z)$, where $H_X$, $J_X$ and $F_X$ ($H_Z$, $J_Z$ and $F_Z$) are the check matrix, logical-operator generator matrix and gauge-operator generator matrix of $X$ ($Z$) operators, respectively. Suppose the code parameters are $[n,k,d]$. Then, $H_X\in\mathbb{F}_2^{r_X\times n}$, $H_Z\in\mathbb{F}_2^{r_Z\times n}$, $J_X,J_Z\in\mathbb{F}_2^{k\times n}$ and $F_X,F_Z\in\mathbb{F}_2^{k_g\times n}$, where $k_g = n-\mathrm{rank}H_X-\mathrm{rank}H_Z-k$ is the number of gauge qubits. These matrices satisfy 
\begin{eqnarray}
\mathrm{ker}H_X &=& \mathrm{rs}H_Z\oplus\mathrm{rs}J_Z\oplus\mathrm{rs}F_Z, \\
\mathrm{ker}H_Z &=& \mathrm{rs}H_X\oplus\mathrm{rs}J_X\oplus\mathrm{rs}F_X, \\
J_XJ_Z^\mathrm{T} &=& E_k, \\
F_XF_Z^\mathrm{T} &=& E_{k_g},
\end{eqnarray}
where $\mathrm{rs}A$ is the row space of the matrix $A$, and $E_k$ is the $k$-dimensional identity matrix. The code distance is 
\begin{eqnarray}
d = \min\{d(H_X,J_X),d(H_Z,J_Z)\},
\end{eqnarray}
where 
\begin{eqnarray}
d(H,J) \equiv \min_{e\in\mathrm{ker}H\,\vert\,Je^\mathrm{T}\neq 0} \vert e\vert,
\end{eqnarray}
where $\vert\bullet\vert$ denotes the Hamming weight. 

Let $X_j$ ($Z_j$) be the $X$ ($Z$) operator of the $j$th qubit. Let $X(v)$ and $Z(v)$ be the $X$ and $Z$ operators of the vector $v\in\mathbb{F}_2^n$, respectively, i.e. 
\begin{eqnarray}
X(v) &\equiv & X_1^{v_1} X_2^{v_2} \cdots X_n^{v_n}, \\
Z(v) &\equiv & Z_1^{v_1} Z_2^{v_2} \cdots Z_n^{v_n}.
\end{eqnarray}
The stabiliser of the code is 
\begin{eqnarray}
\mathcal{S} = \Big\langle X(H_{X;i,\bullet}),Z(H_{Z;j,\bullet})\,\vert\,i=1,2,\ldots,r_X\text{ and }j=1,2,\ldots,r_Z \Big\rangle.
\end{eqnarray}
Here, $A_{i,\bullet}$ ($A_{\bullet,j}$) denotes the $i$th row ($j$th column) of the matrix $A$. The $X$ and $Z$ operators of the $j$th logical qubit are $X(J_{X;j,\bullet})$ and $Z(J_{Z;j,\bullet})$, respectively. Then, 
\begin{eqnarray}
\mathcal{X} &=& \Big\langle X(J_{X;j,\bullet})\,\vert\,j=1,2,\ldots,k \Big\rangle
\end{eqnarray}
and 
\begin{eqnarray}
\mathcal{Z} &=& \Big\langle Z(J_{Z;j,\bullet})\,\vert\,j=1,2,\ldots,k \Big\rangle
\end{eqnarray}
are the groups of $X$ and $Z$ logical operators, respectively. 

{\bf Hypergraph product codes.} Let $H_1\in\mathbb{F}_2^{r_1\times n_1}$ and $H_2\in\mathbb{F}_2^{r_2\times n_2}$ be check matrices of two binary linear codes, respectively. A hypergraph product code generated by $H_1$ and $H_2$ is a quantum code with check matrices 
\begin{eqnarray}
H_X &=& \begin{pmatrix} H_1\otimes E_{n_2} & E_{r_1}\otimes H_2^\mathrm{T} \end{pmatrix}, \\
H_Z &=& \begin{pmatrix} E_{n_1}\otimes H_2 & H_1^\mathrm{T}\otimes E_{r_2} \end{pmatrix}.
\end{eqnarray}

{\bf Repetition code.} A repetition code of length $n$ is a binary linear code with the check matrix 
\begin{eqnarray}
\lambda_n = \begin{pmatrix} E_{n-1} & 0_{n-1,1} \end{pmatrix} + \begin{pmatrix} 0_{n-1,1} & E_{n-1} \end{pmatrix},
\end{eqnarray}
where $0_{a,b}$ is an $a\times b$ zero matrix. Furthermore, $\lambda_{n;\bullet,1:n-1}$ is the matrix generated by deleting the $n$th column from $\lambda_n$. 

For examples, the check matrix of the length-$5$ repetition code is 
\begin{eqnarray}
\lambda_5 = \begin{pmatrix}
1 & 1 & 0 & 0 & 0 \\
0 & 1 & 1 & 0 & 0 \\
0 & 0 & 1 & 1 & 0 \\
0 & 0 & 0 & 1 & 1
\end{pmatrix},
\end{eqnarray}
and 
\begin{eqnarray}
\lambda_{5;\bullet,1:4} = \begin{pmatrix}
1 & 1 & 0 & 0 \\
0 & 1 & 1 & 0 \\
0 & 0 & 1 & 1 \\
0 & 0 & 0 & 1
\end{pmatrix}.
\end{eqnarray}

{\bf Tanner graphs.} We can represent check matrices of a quantum code with a Tanner graph $(\mathcal{B},\mathcal{C}_X,\mathcal{C}_Z,\mathcal{E}_X,\mathcal{E}_Z)$, where $\mathcal{B} = \{1,2,\ldots,n\}$ is the set of bits, $\mathcal{C}_X = \{x_1,x_2,\ldots,x_{r_X}\}$ and $\mathcal{C}_Z = \{z_1,z_2,\ldots,z_{r_Z}\}$ are sets of checks, and $\mathcal{E}_X\subset \mathcal{B}\times \mathcal{C}_X$ and $\mathcal{E}_Z\subset \mathcal{B}\times \mathcal{C}_Z$ are sets of edges. The bipartite graph $(\mathcal{B},\mathcal{C}_X,\mathcal{E}_X)$ is the Tanner graph of the check matrix $H_X$, and the bipartite graph $(\mathcal{B},\mathcal{C}_Z,\mathcal{E}_Z)$ is the Tanner graph of the check matrix $H_Z$. 

Let $(\mathcal{B},\mathcal{C},\mathcal{E})$ be the Tanner graph of a check matrix $H$. We use $\mathcal{B}(\mathcal{E},a)$ to denote the subset of bits that are adjacent to the check $a\in\mathcal{C}$, i.e. 
\begin{eqnarray}
\mathcal{B}(\mathcal{E},a) = \{u\in\mathcal{B} \,\vert\, (u,a)\in\mathcal{E}\}.
\end{eqnarray}
We use $\mathcal{C}(\mathcal{E},u)$ to denote the subset of checks that are adjacent to the bit $u\in\mathcal{B}$, i.e. 
\begin{eqnarray}
\mathcal{C}(\mathcal{E},u) = \{a\in\mathcal{C} \,\vert\, (u,a)\in\mathcal{E}\}.
\end{eqnarray}
We use $\mathcal{E}(u)$ to denote the subset of edges that are incident on the bit $u\in\mathcal{B}$, i.e. 
\begin{eqnarray}
\mathcal{E}(u) = \{(u,a)\in\mathcal{E}\}.
\end{eqnarray}
We use $w_{max}(H)$ to denote the maximum number of non-zero entries in columns and rows of $H$, the maximum vertex degree of $(\mathcal{B},\mathcal{C},\mathcal{E})$. 

{\bf Supports.} We use $\mathcal{Q}(\sigma)\subseteq \mathcal{B}$ to denote the support of the Pauli operator $\sigma$, i.e. $\sigma$ acts non-trivially on and only on qubits in $\mathcal{Q}(\sigma)$. Given a subset of $Z$ logical operators 
\begin{eqnarray}
\Sigma = \{\sigma_1,\sigma_2,\ldots,\sigma_q\} \subseteq \mathcal{Z},
\end{eqnarray}
the crowd number of a qubit $u$ is 
\begin{eqnarray}
cn(\Sigma,u) = \sum_{\sigma\in\Sigma}\vert\mathcal{Q}(\sigma)\cap\{u\}\vert.
\end{eqnarray}
The union of the supports of all the logical operators in $\Sigma$ is 
\begin{eqnarray}
\mathcal{Q}(\Sigma) = \bigcup_{\sigma\in\Sigma}\mathcal{Q}(\sigma).
\end{eqnarray}
A $Z$ logical operator $\tau\in\mathcal{Z}$ is said to be in $\mathcal{Q}(\Sigma)$ if and only if there exist a $Z$ stabiliser operator $Z(h)$ ($h\in\mathrm{rs}H_Z$) and $Z$ gauge operator $Z(f)$ ($f\in\mathrm{rs}F_Z$) such that $\mathcal{Q}[Z(h)Z(f)\tau]\subseteq\mathcal{Q}(\Sigma)$. The $Z$ logical operators in $\mathcal{Q}(\Sigma)$ constitute a group $\mathcal{Z}_N$. Let $k_N$ be the number of independent generators of $\mathcal{Z}_N$, and let $q$ be the number of independent generators of $\langle\Sigma\rangle$. Then, the redundancy number of $\Sigma$ is $rn(\Sigma) = k_N-q$. We can compute the redundancy number according to Algorithm \ref{alg:dressing}, in which the rank of $G_2$ is the redundancy number of $\Sigma$. 

{\bf Standard from.} For a linear code, we say a generator matrix is in the standard form if and only if the matrix is in the form $J = \begin{pmatrix} E & J' \end{pmatrix}$ up to permutations of rows and columns. We can always obtain a generator matrix in the standard form through Gaussian elimination. Similarly, the generator matrices of an $[[n,k,d]]$ CSS code can also be written in the form $J_X = \begin{pmatrix} E_k & 0 & J'_X \end{pmatrix}$ and $J_Z = \begin{pmatrix} E_k & J'_Z & 0 \end{pmatrix}$. 

\section{Glue code}
\label{app:glue_code}

We can simultaneously operate $Z$ logical operators in $\Sigma$ by attaching a sticker to the memory. The sticker is constructed according to the glue code, which is a binary linear code. In this section, we define the glue code in detail. 

\subsection{Compatible glue codes}

\begin{definition}
{\bf Compatible glue code.} Let $(H_X,H_Z,J_X,J_Z,F_X,F_Z)$ be the code of the memory. Let $H_G\in \mathbb{F}_2^{r_G\times n_G}$ be the check matrix of the glue code. The glue code is said to be compatible with the memory if and only if there exists pasting matrices $S\in\mathbb{F}_2^{n_G\times n}$ and $T\in\mathbb{F}_2^{r_X\times r_G}$ such that 
\begin{eqnarray}
H_XS^\mathrm{T} = TH_G.
\label{eq:glue_code}
\end{eqnarray}
\label{def:glue_code}
\end{definition}

As an example, we consider an $X$-operator check matrix in the form 
\begin{eqnarray}
H_X = \begin{pmatrix} H_N & A_X \\ 0_{(r_X-r_N)\times n_N} & B_X \end{pmatrix}.
\label{eq:HXblock}
\end{eqnarray}
Then, the glue code 
\begin{eqnarray}
H_G = \begin{pmatrix} H_N & 0_{r_N\times(n_G-n_N)} \\ A_G & B_G \end{pmatrix}
\label{eq:HGblock}
\end{eqnarray}
is compatible with the memory. By taking pasting matrices 
\begin{eqnarray}
S &=& \begin{pmatrix} E_{n_N} & 0_{n_N\times (n-n_N)} \\ 0_{(n_G-n_N)\times n_N} & 0_{(n_G-n_N)\times (n-n_N)} \end{pmatrix}, \\
T &=& \begin{pmatrix} E_{r_N} & 0_{r_N\times (r_G-r_N)} \\ 0_{(r_X-r_N)\times r_N} & 0_{(r_X-r_N)\times (r_G-r_N)} \end{pmatrix},
\end{eqnarray}
we have 
\begin{eqnarray}
H_XS^\mathrm{T} = TH_G = \begin{pmatrix} H_N & 0_{r_N\times(n_G-n_N)} \\ 0_{(r_X-r_N)\times n_N} & 0_{(r_X-r_N)\times(n_G-n_N)} \end{pmatrix}.
\end{eqnarray}

\begin{lemma}
If the glue code is compatible with the memory, $(\mathrm{ker}H_G)S \subseteq \mathrm{ker}H_X$. 
\label{lem:glue_code}
\end{lemma}

\begin{proof}
For all $u\in \mathrm{ker}H_G$, $TH_Gu^\mathrm{T} = 0$. Then, $H_XS^\mathrm{T}u^\mathrm{T} = 0$, i.e. $uS\in \mathrm{ker}H_X$. 
\end{proof}

\subsection{Glue codes devised for $\Sigma$}

The operation realised by a sticker is determined by the glue code. To operate logical operators in $\Sigma$, we need to construct a compatible glue code, and the code also needs to meet $\Sigma$. 

\begin{definition}
{\bf Devised glue codes.} Let $v_1,v_2,\ldots,v_q \in \mathrm{rs}(J_Z)$ be vectors corresponding to the operator set $\Sigma$, i.e. $Z(v_i) = \sigma_i$ for $i = 1,2,\ldots,q$. Let $H_G$ be the check matrix of a glue code that is compatible with the memory. The glue code is said to be {\bf coarsely devised} for $\Sigma$ if and only if 
\begin{eqnarray}
\mathrm{span}(v_1,v_2,\ldots,v_q) \subseteq (\mathrm{ker}H_G)S.
\end{eqnarray}
The glue code is said to be {\bf finely devised} for $\Sigma$ if and only if there exists $u_1,u_2,\ldots\in\mathrm{rs}H_Z\oplus\mathrm{rs}F_Z$ such that 
\begin{eqnarray}
\mathrm{span}(v_1,v_2,\ldots,v_q,u_1,u_2,\ldots) = (\mathrm{ker}H_G)S.
\end{eqnarray}
\label{def:devised}
\end{definition}

We will give a systemic approach for constructing devised glue codes in Sec. \ref{app:sticking}. 

Vectors $\{v_1,v_2,\ldots,v_q\}$ span a subspace $\mathrm{rs}J_{Z,A}$, where 
\begin{eqnarray}
J_{Z,A} = \begin{pmatrix} v_1 \\ v_2 \\ \vdots \\ v_q \end{pmatrix}.
\end{eqnarray}
Then, all logical operators in $Z(\mathrm{rs}J_{Z,A}) = \langle\Sigma\rangle$ are actively operated by the sticker. Without loss of generality, we suppose that $\{v_1,v_2,\ldots,v_q\}$ are linearly independent. Then, we can find a basis of $\mathrm{rs}J_Z$ by extending $\{v_1,v_2,\ldots,v_q\}$, denoted by 
\begin{eqnarray}
\{v_1,v_2,\ldots,v_q\}\cup\{v_{q+1},v_{q+2},\ldots,v_k\}.
\end{eqnarray}
Vectors $\{v_{q+1},v_{q+2},\ldots,v_k\}$ span the complementary subspace $\mathrm{rs}J_{Z,C}$, where 
\begin{eqnarray}
J_{Z,C} = \begin{pmatrix} v_{q+1} \\ v_{q+2} \\ \vdots \\ v_k \end{pmatrix}.
\end{eqnarray}
Since $\mathrm{rs}J_Z = \mathrm{rs}J_{Z,A}\oplus\mathrm{rs}J_{Z,C}$, there exist a full rank matrix $\bar{J}_Z\in\mathbb{F}_2^{k\times k}$ such that 
\begin{eqnarray}
\begin{pmatrix} J_{Z,A} \\ J_{Z,C} \end{pmatrix} &=& \bar{J}_ZJ_Z.
\end{eqnarray}
Let $\bar{J}_X = \bar{J}_Z^{-1}$. We have matrices $\bar{J}_{X,A}\in\mathbb{F}_2^{q\times n}$ and $\bar{J}_{X,C}\in\mathbb{F}_2^{(k-q)\times n}$ defined according to 
\begin{eqnarray}
\begin{pmatrix} J_{X,A} \\ J_{X,C} \end{pmatrix} &=& \bar{J}_XJ_X.
\end{eqnarray}
Then, they satisfy $\mathrm{rs}J_X = \mathrm{rs}J_{X,A}\oplus\mathrm{rs}J_{X,C}$, $J_{X,A}J_{Z,A}^\mathrm{T} = E_q$, $J_{X,C}J_{Z,C}^\mathrm{T} = E_{k-q}$, and $J_{X,A}J_{Z,C}^\mathrm{T} = J_{X,C}J_{Z,A}^\mathrm{T} = 0$. 

\begin{lemma}
If the glue code is finely devised for $\Sigma$, there exists $\gamma\in \mathbb{F}_2^{(k-q)\times r_G}$ such that $J_{X,C}S^\mathrm{T} = \gamma H_G$. 
\label{lem:JXC}
\end{lemma}

\begin{proof}
According to Definition \ref{def:devised}, there exists a matrix $K_Z$ and a generator matrix of the glue code $\mathrm{ker}H_G$, denoted by $G$, satisfying $\mathrm{rs}K_Z\subseteq\mathrm{rs}H_Z\oplus\mathrm{rs}F_Z$ and 
\begin{eqnarray}
\begin{pmatrix} J_{Z,A} \\ K_Z \end{pmatrix} &=& GS.
\end{eqnarray}
Then, $J_{X,C}S^\mathrm{T}G^\mathrm{T} = 0$. Therefore, $\mathrm{rs}(J_{X,C}S^\mathrm{T}) \subseteq \mathrm{ker}G = \mathrm{rs}H_G$. 
\end{proof}

\section{Stickers}
\label{app:stickers}

In this section, we define the two types of stickers, measurement stickers and branch stickers.

\begin{definition}
{\bf Measurement stickers.} A measurement sticker is a hypergraph product code generated by check matrices $H_G$ and $\lambda_{d_R}^\mathrm{T}$, where $H_G$ is the check matrix of the glue code. The $X$- and $Z$-operator check matrices of the measurement sticker are 
\begin{eqnarray}
H^M_X &=& \begin{pmatrix}  E_{d_R-1}\otimes H_G& \lambda_{d_R}\otimes  E_{r_G}\end{pmatrix}, \\
H^M_Z &=& \begin{pmatrix}\lambda_{d_R}^\mathrm{T}\otimes  E_{n_G} &E_{d_R}  \otimes H_G^\mathrm{T}\end{pmatrix}.
\end{eqnarray}
\end{definition}

For the convenience of subsequent discussions, we expand measurement-sticker check matrices in the form 
\begin{eqnarray}
H^M_X &=& \begin{pmatrix}
H_G & 0 & \cdots & 0 & 0 & E_{r_G} & E_{r_G} & 0 & \cdots & 0 & 0 & 0 \\
0 & H_G & \cdots & 0 & 0 & 0 & E_{r_G} & E_{r_G} & \cdots & 0 & 0 & 0 \\
\vdots & \vdots & \ddots & \vdots & \vdots & \vdots & \vdots & \vdots & \ddots & \vdots & \vdots & \vdots \\
0 & 0 & \cdots & H_G & 0 & 0 & 0 & 0 & \cdots & E_{r_G} & E_{r_G} & 0 \\
0 & 0 & \cdots & 0 & H_G & 0 & 0 & 0 & \cdots & 0 & E_{r_G} & E_{r_G}
\end{pmatrix},
\label{eq:HMX}
\end{eqnarray}
and 
\begin{eqnarray}
H^M_Z &=& \begin{pmatrix}
E_{n_G} & 0 & \cdots & 0 & 0 & H_G^\mathrm{T} & 0 & 0 & \cdots & 0 & 0 & 0 \\
E_{n_G} & E_{n_G} & \cdots & 0 & 0 & 0 & H_G^\mathrm{T} & 0 & \cdots & 0 & 0 & 0 \\
0 & E_{n_G} & \cdots & 0 & 0 & 0 & 0 & H_G^\mathrm{T} & \cdots & 0 & 0 & 0 \\
\vdots & \vdots & \ddots & \vdots & \vdots & \vdots & \vdots & \vdots & \ddots & \vdots & \vdots & \vdots \\
0 & 0 & \cdots & E_{n_G} & 0 & 0 & 0 & 0 & \cdots & H_G^\mathrm{T} & 0 & 0 \\
0 & 0 & \cdots & E_{n_G} & E_{n_G} & 0 & 0 & 0 & \cdots & 0 & H_G^\mathrm{T} & 0 \\
0 & 0 & \cdots & 0 & E_{n_G} & 0 & 0 & 0 & \cdots & 0 & 0 & H_G^\mathrm{T}
\end{pmatrix}.
\label{eq:HMZ}
\end{eqnarray}

\begin{definition}
{\bf Branch stickers.} A branch sticker is a hypergraph product code generated by check matrices $H_G$ and $\lambda_{d_R;\bullet,1:d_R-1}^\mathrm{T}$, where $H_G$ is the check matrix of the glue code. The $X$- and $Z$-operator check matrices of the branch sticker are 
\begin{eqnarray}
{H^B_X }&=& \begin{pmatrix} E_{d_R-1} \otimes H_G& \lambda_{d_R;\bullet,1:d_R-1} \otimes  \E_{r_G} \end{pmatrix}, \\
H^B_Z &=& \begin{pmatrix}  \lambda_{d_R;\bullet,1:d_R-1}^\mathrm{T} \otimes E_{d_R-1}& E_{n_G} \otimes H_G^\mathrm{T}\end{pmatrix}.
\end{eqnarray} 
\end{definition}

For the convenience of subsequent discussions, we expand measurement-sticker check matrices in the form 
\begin{eqnarray}
H^B_X &=& \begin{pmatrix}
H_G & 0 & \cdots & 0 & 0 & E_{r_G} & E_{r_G} & 0 & \cdots & 0 & 0 \\
0 & H_G & \cdots & 0 & 0 & 0 & E_{r_G} & E_{r_G} & \cdots & 0 & 0 \\
\vdots & \vdots & \ddots & \vdots & \vdots & \vdots & \vdots & \vdots & \ddots & \vdots & \vdots \\
0 & 0 & \cdots & H_G & 0 & 0 & 0 & 0 & \cdots & E_{r_G} & E_{r_G} \\
0 & 0 & \cdots & 0 & H_G & 0 & 0 & 0 & \cdots & 0 & E_{r_G}
\end{pmatrix},
\end{eqnarray}
and 
\begin{eqnarray}
H^B_Z &=& \begin{pmatrix}
E_{n_G} & 0 & \cdots & 0 & 0 & H_G^\mathrm{T} & 0 & 0 & \cdots & 0 & 0 \\
E_{n_G} & E_{n_G} & \cdots & 0 & 0 & 0 & H_G^\mathrm{T} & 0 & \cdots & 0 & 0 \\
0 & E_{n_G} & \cdots & 0 & 0 & 0 & 0 & H_G^\mathrm{T} & \cdots & 0 & 0 \\
\vdots & \vdots & \ddots & \vdots & \vdots & \vdots & \vdots & \vdots & \ddots & \vdots & \vdots \\
0 & 0 & \cdots & E_{n_G} & 0 & 0 & 0 & 0 & \cdots & H_G^\mathrm{T} & 0 \\
0 & 0 & \cdots & E_{n_G} & E_{n_G} & 0 & 0 & 0 & \cdots & 0 & H_G^\mathrm{T}
\end{pmatrix}.
\end{eqnarray}

We can find that $H^B_X$ can be generated by deleting the last column from $H^M_X$ in Eq. (\ref{eq:HMX}), and $H^B_Z$ can be generated by deleting the last column and last row from $H^M_Z$ in Eq. (\ref{eq:HMZ}). 

\section{Deformed codes}
\label{app:deformed_codes}

When the glue code of a sticker is compatible with the memory, we can attach the sticker to the memory and generate a deformed code. In this section, we define the deformed codes and their logical operators. We need to define the logical operators because the deformed codes are subsystem codes in general. With logical operators defined, we analyse distances of deformed codes. 

\subsection{Definitions}

\begin{definition}
{\bf Measurement-sticker deformed codes.} Let $(H_X,H_Z)$ be check matrices of the memory. Let $H_G$ be the check matrix of the glue code. Suppose the glue code is compatible with the memory and finely devised for a set of $Z$ logical operators $\Sigma$. The deformed code is generated by attaching the corresponding measurement sticker to the memory, and its $X$- and $Z$-operator check matrices are 
\begin{eqnarray}
H^{M-M}_X &=& \begin{pmatrix}
H_X & 0 & 0 & \cdots & 0 & 0 & T & 0 & 0 & \cdots & 0 & 0 & 0 \\
0 & H_G & 0 & \cdots & 0 & 0 & E_{r_G} & E_{r_G} & 0 & \cdots & 0 & 0 & 0 \\
0 & 0 & H_G & \cdots & 0 & 0 & 0 & E_{r_G} & E_{r_G} & \cdots & 0 & 0 & 0 \\
\vdots & \vdots & \vdots & \ddots & \vdots & \vdots & \vdots & \vdots & \vdots & \ddots & \vdots & \vdots & \vdots \\
0 & 0 & 0 & \cdots & H_G & 0 & 0 & 0 & 0 & \cdots & E_{r_G} & E_{r_G} & 0 \\
0 & 0 & 0 & \cdots & 0 & H_G & 0 & 0 & 0 & \cdots & 0 & E_{r_G} & E_{r_G}
\end{pmatrix},
\end{eqnarray}
and 
\begin{eqnarray}
H^{M-M}_Z &=& \begin{pmatrix}
H_Z & 0 & 0 & \cdots & 0 & 0 & 0 & 0 & 0 & \cdots & 0 & 0 & 0 \\
S & E_{n_G} & 0 & \cdots & 0 & 0 & H_G^\mathrm{T} & 0 & 0 & \cdots & 0 & 0 & 0 \\
0 & E_{n_G} & E_{n_G} & \cdots & 0 & 0 & 0 & H_G^\mathrm{T} & 0 & \cdots & 0 & 0 & 0 \\
0 & 0 & E_{n_G} & \cdots & 0 & 0 & 0 & 0 & H_G^\mathrm{T} & \cdots & 0 & 0 & 0 \\
\vdots & \vdots & \vdots & \ddots & \vdots & \vdots & \vdots & \vdots & \vdots & \ddots & \vdots & \vdots & \vdots \\
0 & 0 & 0 & \cdots & E_{n_G} & 0 & 0 & 0 & 0 & \cdots & H_G^\mathrm{T} & 0 & 0 \\
0 & 0 & 0 & \cdots & E_{n_G} & E_{n_G} & 0 & 0 & 0 & \cdots & 0 & H_G^\mathrm{T} & 0 \\
0 & 0 & 0 & \cdots & 0 & E_{n_G} & 0 & 0 & 0 & \cdots & 0 & 0 & H_G^\mathrm{T}
\end{pmatrix}.
\end{eqnarray}
\label{def:measurement_deformed}
\end{definition}

\begin{definition}
{\bf Branch-sticker deformed codes.} Let $(H_X,H_Z)$ be check matrices of the memory. Let $H_G$ be the check matrix of the glue code. Suppose the glue code is compatible with the memory. The deformed code is generated by attaching the corresponding branch sticker to the memory, and its $X$- and $Z$-operator check matrices are 
\begin{eqnarray}
H^{M-B}_X &=& \begin{pmatrix}
H_X & 0 & 0 & \cdots & 0 & 0 & T & 0 & 0 & \cdots & 0 & 0 \\
0 & H_G & 0 & \cdots & 0 & 0 & E_{r_G} & E_{r_G} & 0 & \cdots & 0 & 0 \\
0 & 0 & H_G & \cdots & 0 & 0 & 0 & E_{r_G} & E_{r_G} & \cdots & 0 & 0 \\
\vdots & \vdots & \vdots & \ddots & \vdots & \vdots & \vdots & \vdots & \vdots & \ddots & \vdots & \vdots \\
0 & 0 & 0 & \cdots & H_G & 0 & 0 & 0 & 0 & \cdots & E_{r_G} & E_{r_G} \\
0 & 0 & 0 & \cdots & 0 & H_G & 0 & 0 & 0 & \cdots & 0 & E_{r_G}
\end{pmatrix},
\end{eqnarray}
and 
\begin{eqnarray}
H^{M-B}_Z &=& \begin{pmatrix}
H_Z & 0 & 0 & \cdots & 0 & 0 & 0 & 0 & 0 & \cdots & 0 & 0 \\
S & E_{n_G} & 0 & \cdots & 0 & 0 & H_G^\mathrm{T} & 0 & 0 & \cdots & 0 & 0 \\
0 & E_{n_G} & E_{n_G} & \cdots & 0 & 0 & 0 & H_G^\mathrm{T} & 0 & \cdots & 0 & 0 \\
0 & 0 & E_{n_G} & \cdots & 0 & 0 & 0 & 0 & H_G^\mathrm{T} & \cdots & 0 & 0 \\
\vdots & \vdots & \vdots & \ddots & \vdots & \vdots & \vdots & \vdots & \vdots & \ddots & \vdots & \vdots \\
0 & 0 & 0 & \cdots & E_{n_G} & 0 & 0 & 0 & 0 & \cdots & H_G^\mathrm{T} & 0 \\
0 & 0 & 0 & \cdots & E_{n_G} & E_{n_G} & 0 & 0 & 0 & \cdots & 0 & H_G^\mathrm{T}
\end{pmatrix}.
\end{eqnarray}
\label{def:branch_deformed}
\end{definition}

\begin{proposition}
As a consequence of Eq. (\ref{eq:glue_code}) in Definition \ref{def:glue_code}, check matrices of deformed codes are compatible, i.e. $H^{M-M}_X{H^{M-M}_Z}^\mathrm{T} = 0$ and $H^{M-B}_X{H^{M-B}_Z}^\mathrm{T} = 0$. 
\end{proposition}

\subsection{Logical operators}

\begin{definition}
{\bf Logical operators of measurement-sticker deformed codes.} For a measurement-sticker deformed code as defined in Definition \ref{def:measurement_deformed}, its $X$- and $Z$-operator generator matrices are 
\begin{eqnarray}
J^{M-M}_X &=& \begin{pmatrix}
J_{X,C} & J_{X,C}S^\mathrm{T} & J_{X,C}S^\mathrm{T} & \cdots & J_{X,C}S^\mathrm{T} & J_{X,C}S^\mathrm{T} & 0 & 0 & 0 & \cdots & 0 & 0 & \gamma
\end{pmatrix},
\end{eqnarray}
and 
\begin{eqnarray}
J^{M-M}_Z &=& \begin{pmatrix}
J_{Z,C} & 0 & 0 & \cdots & 0 & 0 & 0 & 0 & 0 & \cdots & 0 & 0 & 0
\end{pmatrix}.
\end{eqnarray}
\end{definition}

\begin{definition}
{\bf Logical operators of branch-sticker deformed codes.} For a branch-sticker deformed code as defined in Definition \ref{def:branch_deformed}, its $X$- and $Z$-operator generator matrices are 
\begin{eqnarray}
J^{M-B}_X &=& \begin{pmatrix}
J_{X,A} & J_{X,A}S^\mathrm{T} & J_{X,A}S^\mathrm{T} & \cdots & J_{X,A}S^\mathrm{T} & J_{X,A}S^\mathrm{T} & 0 & 0 & 0 & \cdots & 0 & 0 \\
J_{X,C} & J_{X,C}S^\mathrm{T} & J_{X,C}S^\mathrm{T} & \cdots & J_{X,C}S^\mathrm{T} & J_{X,C}S^\mathrm{T} & 0 & 0 & 0 & \cdots & 0 & 0
\end{pmatrix},
\end{eqnarray}
and 
\begin{eqnarray}
J^{M-B}_Z &=& \begin{pmatrix}
J_{Z,A} & 0 & 0 & \cdots & 0 & 0 & 0 & 0 & 0 & \cdots & 0 & 0 \\
J_{Z,C} & 0 & 0 & \cdots & 0 & 0 & 0 & 0 & 0 & \cdots & 0 & 0
\end{pmatrix}.
\end{eqnarray}
\end{definition}

\begin{proposition}
As a consequence of Lemma \ref{lem:JXC}, generator matrices of measurement-sticker deformed codes are valid, i.e. $H^{M-M}_X{J^{M-M}_Z}^\mathrm{T} = H^{M-M}_Z{J^{M-M}_X}^\mathrm{T} = 0$ and $J^{M-M}_X{J^{M-M}_Z}^\mathrm{T} = E_{k-q}$. Generator matrices of branch-sticker deformed codes are valid, i.e. $H^{M-B}_X{J^{M-B}_Z}^\mathrm{T} = H^{M-B}_Z{J^{M-B}_X}^\mathrm{T} = 0$ and $J^{M-B}_X{J^{M-B}_Z}^\mathrm{T} = E_k$. 
\end{proposition}

\subsection{Code distances}

\begin{lemma}
Let the distance of the memory be $d$. The distance of the measurement-sticker deformed code has the lower bound 
\begin{eqnarray}
d^{M-M} &\geq & \min\{d/\vert S\vert,d_R\}.
\end{eqnarray}
Here, $\vert S\vert$ denotes the matrix norm induced by the Hamming weight, where $S$ acts on the vector from the right side. 
\label{lem:dis_mea}
\end{lemma}

\begin{proof}
{\bf $X$-operator distance.} We prove the distance lower bound by contradiction. Suppose that there exists a $Z$ logical error $e$, i.e. $H^{M-M}_Xe^\mathrm{T} = 0$ and $J^{M-M}_Xe^\mathrm{T} \neq 0$, but its weight is $\vert e\vert < \min\{d/\vert S\vert,d_R\}$. Let the error be 
\begin{eqnarray}
e = \begin{pmatrix}
u_0 & u_1 & u_2 & \cdots & u_{d_R-2} & u_{d_R-1} & v_1 & v_2 & v_3 & \cdots & v_{d_R-1} & v_{d_R}
\end{pmatrix}.
\label{eq:error}
\end{eqnarray}
As a consequence of $H^{M-M}_Xe^\mathrm{T} = 0$, the following equations hold, 
\begin{eqnarray}
H_Xu_0^\mathrm{T} &=& Tv_1^\mathrm{T}, \label{eq:u0} \\
H_Gu_j^\mathrm{T} &=& v_j^\mathrm{T}+v_{j+1}^\mathrm{T}, \label{eq:uj}
\end{eqnarray}
where $j = 1,2,\ldots,d_R-1$. Because $\vert e\vert < d_R$, one of $v_1,v_2,\ldots,v_{d_R}$ must be zero. Suppose $v_l = 0$. According to Eq. (\ref{eq:uj}), 
\begin{eqnarray}
H_G\sum_{j=1}^{l-1}u_j^\mathrm{T} &=& v_1^\mathrm{T}.
\end{eqnarray}
Substitute $v_1^\mathrm{T}$ into Eq. (\ref{eq:u0}), we have 
\begin{eqnarray}
H_Xu_0^\mathrm{T} = TH_G\sum_{j=1}^{l-1}u_j^\mathrm{T} = H_XS^\mathrm{T}\sum_{j=1}^{l-1}u_j^\mathrm{T}.
\end{eqnarray}
Therefore, 
\begin{eqnarray}
H_X\left(u_0^\mathrm{T} + S^\mathrm{T}\sum_{j=1}^{l-1}u_j^\mathrm{T}\right) = 0.
\label{eq:HX}
\end{eqnarray}

As a consequence of $J^{M-M}_Xe^\mathrm{T} \neq 0$, 
\begin{eqnarray}
J_{X,C}u_0^\mathrm{T} + J_{X,C}S^\mathrm{T}\sum_{j=1}^{d_R-1}u_j^\mathrm{T} + \gamma v_{d_R}^\mathrm{T} \neq 0.
\end{eqnarray}
We consider the term 
\begin{eqnarray}
x = J_{X,C}S^\mathrm{T}\sum_{j=l}^{d_R-1}u_j^\mathrm{T} + \gamma v_{d_R}^\mathrm{T}.
\end{eqnarray}
According to Lemma \ref{lem:JXC}, 
\begin{eqnarray}
x = \gamma \left(H_G\sum_{j=l}^{d_R-1}u_j^\mathrm{T} +  v_{d_R}^\mathrm{T}\right).
\end{eqnarray}
Because of Eq. (\ref{eq:uj}) and $v_l = 0$, $x = 0$. Therefore, 
\begin{eqnarray}
J_{X,C}\left(u_0^\mathrm{T} + S^\mathrm{T}\sum_{j=1}^{l-1}u_j^\mathrm{T}\right) \neq 0.
\label{eq:JX}
\end{eqnarray}

According to Eqs. (\ref{eq:HX}) and (\ref{eq:JX}), the error 
\begin{eqnarray}
u = u_0 + \left(\sum_{j=1}^{l-1}u_j\right)S
\end{eqnarray}
is a logical error of the memory, i.e. $\vert u\vert\geq d$. Using inequalities 
\begin{eqnarray}
\vert u\vert &\leq& \vert u_0\vert + \left(\sum_{j=1}^{l-1}\vert u_j\vert\right)\vert S\vert, \\
\vert e\vert &\geq& \vert u_0\vert + \left(\sum_{j=1}^{l-1}\vert u_j\vert\right),
\end{eqnarray}
and $\vert S\vert \geq 1$ (when $S\neq 0$), we have $\vert e\vert \geq \vert u\vert/\vert S\vert \geq d/\vert S\vert$, which is a contradiction. 

{\bf $Z$-operator distance.} Let $e$ be an $X$ logical error, i.e. $H^{M-M}_Ze^\mathrm{T} = 0$ and $J^{M-M}_Ze^\mathrm{T} \neq 0$. Expressing $e$ in the form of Eq. (\ref{eq:error}), we have 
\begin{eqnarray}
H_Zu_0^\mathrm{T} &=&0, \\
J_{Z,C}u_0^\mathrm{T} &\neq& 0.
\end{eqnarray}
Therefore, $u_0$ is a logical error of the memory, i.e. $\vert u_0\vert\geq d$. Then, $\vert e\vert \geq \vert u_0\vert \geq d$. 
\end{proof}

\begin{lemma}
Let the distance of the memory be $d$. The distance of the branch-sticker deformed code has the lower bound 
\begin{eqnarray}
d^{M-M} &\geq & d/\vert S\vert.
\end{eqnarray}
\label{lem:dis_bran}
\end{lemma}

\begin{proof}
{\bf $X$-operator distance.} The proof is similar to Lemma \ref{lem:dis_mea}. We remove the entry $v_{d_R}$ from the expression of the error $e$ in Eq. (\ref{eq:error}) and take $v_{d_R} = 0$ and $l = d_R$ in the equations. Because we always have $v_{d_R} = 0$, the existence of a zero-valued $v_j$ entry is independent of $d_R$, i.e. the distance lower bound is independent of $d_R$. 

{\bf $Z$-operator distance.} The proof is the same as Lemma \ref{lem:dis_mea}. 
\end{proof}

\section{Relations between memory operators and deformed-code operators}

To establish the relations between memory operators and deformed-code operators, we introduce three matrices representing supports of the memory on the measurement-sticker deformed code, the memory on the branch-sticker deformed code and the open boundary on the branch sticker, respectively. They are 
\begin{eqnarray}
P_{M-M} &=& \begin{pmatrix}
E_n & 0 & 0 & \cdots & 0 & 0 & 0 & 0 & 0 & \cdots & 0 & 0 & 0
\end{pmatrix},
\end{eqnarray}
\begin{eqnarray}
P_{M-B} &=& \begin{pmatrix}
E_n & 0 & 0 & \cdots & 0 & 0 & 0 & 0 & 0 & \cdots & 0 & 0
\end{pmatrix}
\end{eqnarray},
and 
\begin{eqnarray}
P_{ob} &=& \begin{pmatrix}
0 & 0 & 0 & \cdots & 0 & E_{n_G} & 0 & 0 & 0 & \cdots & 0 & 0
\end{pmatrix}.
\end{eqnarray}
The matrices $P_{M-M}^\mathrm{T}P_{M-M}$, $P_{M-B}^\mathrm{T}P_{M-B}$ and $P_{ob}^\mathrm{T}P_{ob}$ are projections onto supports of the memory on the measurement-sticker deformed code, the memory on the branch-sticker deformed code and the open boundary on the branch sticker, respectively. 

\begin{lemma}
For a measurement-sticker deformed code, 
\begin{eqnarray}
\mathrm{rs}H_X &=& \mathrm{rs}(H^{M-M}_XP_{M-M}^\mathrm{T}), \label{eq:rel_mea_HX} \\
\mathrm{rs}(H_ZP_{M-M}) &\subseteq& \mathrm{rs}H^{M-M}_Z, \label{eq:rel_mea_HZ} \\
\mathrm{rs}J_{X,C} &=& \mathrm{rs}(J^{M-M}_XP_{M-M}^\mathrm{T}), \label{eq:rel_mea_JXC} \\
\mathrm{rs}(J_{Z,C}P_{M-M}) &=& \mathrm{rs}J^{M-M}_Z, \label{eq:rel_mea_JZC} \\
\mathrm{rs}(J_{Z,A}P_{M-M}) &\subseteq& \mathrm{rs}H^{M-M}_Z \label{eq:rel_mea_JZA}.
\end{eqnarray}
\label{lem:rel_mea}
\end{lemma}

\begin{proof}
The matrix $H^{M-M}_XP_{M-M}^\mathrm{T}$ is the first column of $H^{M-M}_X$, in which the first row is $H_X$, and all other rows are zero. Therefore, Eq. (\ref{eq:rel_mea_HX}) holds. The matrix $H_ZP_{M-M}$ is the first row of $H^{M-M}_Z$. Therefore, Eq. (\ref{eq:rel_mea_HZ}) holds. 

Because $J_{X,C} = J^{M-M}_XP_{M-M}^\mathrm{T}$ and $J_{Z,C}P_{M-M} = J^{M-M}_Z$, Eqs. (\ref{eq:rel_mea_JXC}) and (\ref{eq:rel_mea_JZC}) hold. 

According to the definition of devised glue codes, there exists $J_G$ such that $\mathrm{rs}J_G\subseteq \mathrm{ker}H_G$ and $J_{Z,A} = J_GS$. 
By taking 
\begin{eqnarray}
\beta &=& \begin{pmatrix}
0 & J_G & J_G & J_G & \cdots & J_G & J_G & J_G
\end{pmatrix},
\label{eq:beta}
\end{eqnarray}
we have 
\begin{eqnarray}
J_{Z,A}P_{M-M} = \beta H^{M-M}_Z. 
\end{eqnarray}
Therefore, Eq. (\ref{eq:rel_mea_JZA}) holds. 
\end{proof}

\begin{lemma}
For a branch-sticker deformed code, 
\begin{eqnarray}
\mathrm{rs}H_X &=& \mathrm{rs}(H^{M-B}_XP_{M-B}^\mathrm{T}), \label{eq:rel_bran_HX} \\
\mathrm{rs}(H_ZP_{M-B}) &\subseteq& \mathrm{rs}H^{M-B}_Z, \label{eq:rel_bran_HZ} \\
\mathrm{rs}J_X &=& \mathrm{rs}(J^{M-B}_ZP_{M-B}^\mathrm{T}), \label{eq:rel_bran_JX} \\
\mathrm{rs}(J_ZP_{M-B}) &=& \mathrm{rs}J^{M-B}_Z, \label{eq:rel_bran_JZ} \\
\mathrm{rs}(J_{Z,A}P_{M-B}+J_GP_{ob}) &\subseteq& \mathrm{rs}H^{M-B}_Z \label{eq:rel_bran_JZA}.
\end{eqnarray}
Here, $J_G$ satisfies $\mathrm{rs}J_G\subseteq \mathrm{ker}H_G$ and $J_{Z,A} = J_GS$. 
\label{lem:rel_bran}
\end{lemma}

\begin{proof}
Eqs. (\ref{eq:rel_bran_HX}), (\ref{eq:rel_bran_HZ}), (\ref{eq:rel_bran_JX}) and (\ref{eq:rel_bran_JZ}) are proved as the same as Eqs. (\ref{eq:rel_mea_HX}), (\ref{eq:rel_mea_HZ}), (\ref{eq:rel_mea_JXC}) and (\ref{eq:rel_mea_JZC}) by noticing that 
\begin{eqnarray}
\mathrm{rs}\begin{pmatrix} J_{X,A} \\ J_{X,C} \end{pmatrix} &=& \mathrm{rs}J_X, \\
\mathrm{rs}\begin{pmatrix} J_{Z,A} \\ J_{Z,C} \end{pmatrix} &=& \mathrm{rs}J_Z.
\end{eqnarray}

By taking $\beta$ in the form of Eq. (\ref{eq:beta}) with the last entry removed, 
we have 
\begin{eqnarray}
J_{Z,A}P_{M-B}+J_GP_{ob} = \beta H^{M-B}_Z. 
\end{eqnarray}
Therefore, Eq. (\ref{eq:rel_bran_JZA}) holds. 
\end{proof}

\section{Theorem of the generalised lattice surgery}
\label{app:lattice_surgery}

\begin{theorem}
Let $\mathcal{S}$ and $\mathcal{S}_{dc}$ be stabilisers of the memory and deformed code, respectively. Let $\mathcal{X}$ and $\mathcal{X}_{dc}$ ($\mathcal{Z}$ and $\mathcal{Z}_{dc}$) be groups of $X$ ($Z$) logical operators of the memory and deformed code, respectively. Let the distance of the memory be $d$. The following statements hold, 
\begin{itemize}
\item[i)] For each $X$ stabiliser operator of the memory (deformed code) $g\in \mathcal{S}$ ($g_{dc}\in \mathcal{S}_{dc}$), there exists an $X$ stabiliser operator of the deformed code (memory) $g_{dc}\in \mathcal{S}_{dc}$ ($g\in \mathcal{S}$) such that the support of $gg_{dc}$ is on the sticker; 
\item[ii)] For each $Z$ stabiliser operator of the memory $g\in \mathcal{S}$, there exists a $Z$ stabiliser operator of the deformed code $g_{dc}\in \mathcal{S}_{dc}$ such that $g = g_{dc}$. 
\end{itemize}
If the deformed code is generated by a measurement sticker, 
\begin{itemize}
\item[iii)] For each $X$ logical operator of the memory $\tau\in \mathcal{X}$ that commutes with operators in $\langle\Sigma\rangle$, there exists an $X$ logical operator of the deformed code $\tau_{dc}\in \mathcal{X}_{dc}$ such that the support of $\tau\tau_{dc}$ is on the sticker; 
\item[iv)] For each $Z$ logical operator of the memory $\tau\in \mathcal{Z}$, there exists a $Z$ logical operator of the deformed code $\tau_{dc}\in \mathcal{Z}_{dc}$ and a $Z$ stabiliser operator of the deformed code $g_{dc}\in \mathcal{S}_{dc}$ such that $\tau = \tau_{dc}g_{dc}$; 
\item[v)] For each $Z$ logical operator of the memory $\tau\in \langle\Sigma\rangle$, there exists a $Z$ stabiliser operator of the deformed code $g_{dc}\in \mathcal{S}_{dc}$ such that $\tau = g_{dc}$; 
\item[vi)] The distance of the deformed code is $d_{dc}\geq \min\{d/\vert S\vert,d_R\}$. 
\end{itemize}
If the deformed code is generated by a branch sticker, 
\begin{itemize}
\item[iii')] For each $X$ logical operator of the memory $\tau\in \mathcal{X}$, there exists an $X$ logical operator of the deformed code $\tau_{dc}\in \mathcal{X}_{dc}$ such that the support of $\tau\tau_{dc}$ is on the sticker; 
\item[iv')] For each $Z$ logical operator of the memory $\tau\in \mathcal{Z}$, there exists a $Z$ logical operator of the deformed code $\tau_{dc}\in \mathcal{Z}_{dc}$ such that $\tau = \tau_{dc}$; 
\item[v')] For each $Z$ logical operator of the memory $\tau\in \langle\Sigma\rangle$, there exists a $Z$ stabiliser operator of the deformed code $g_{dc}\in \mathcal{S}_{dc}$ such that the support of $\tau g_{dc}$ is on the open boundary of the branch sticker; 
\item[vi')] The distance of the deformed code is $d_{dc}\geq d/\vert S\vert$. 
\end{itemize}
\label{the:GLS}
\end{theorem}

\begin{proof}
We have proved every piece of the theorem in Lemmas \ref{lem:dis_mea}, \ref{lem:dis_bran}, \ref{lem:rel_mea} and \ref{lem:rel_bran}. Here, we only need to relate these lemmas to statements in the theorem. 

Let $\mathcal{S}^X$ and $\mathcal{S}^Z$ ($\mathcal{S}_{dc}^X$ and $\mathcal{S}_{dc}^Z$) be the sets of $X$ and $Z$ stabiliser operators of the memory (deformed code), respectively. These operator sets are related to check and generator matrices through $\mathcal{S}^X = X(\mathrm{rs}H_X)$, $\mathcal{S}^Z = Z(\mathrm{rs}H_Z)$, $\mathcal{S}_{dc}^X = X(\mathrm{rs}H^{M-\alpha}_X)$ and $\mathcal{S}_{dc}^Z = Z(\mathrm{rs}H^{M-\alpha}_Z)$, where $\alpha = M,D$. Therefore, statements i) and ii) are consequences of Eqs. (\ref{eq:rel_mea_HX}), (\ref{eq:rel_mea_HZ}), (\ref{eq:rel_bran_HX}) and (\ref{eq:rel_bran_HZ}). 

The set of memory $X$ logical operators that commute with operators in $\Sigma$ is $X(\mathrm{rs}J_{X,C})$. Therefore, the statement iii) is a consequence of Eq. (\ref{eq:rel_mea_JXC}). 

According to Eqs. (\ref{eq:rel_mea_JZC}) and (\ref{eq:rel_mea_JZA}), 
\begin{eqnarray}
\mathrm{rs}(J_ZP_{M-M}) &\subseteq & \mathrm{rs}\begin{pmatrix} H^{M-M}_Z \\ J^{M-M}_Z \end{pmatrix}.
\end{eqnarray}
Because $\mathcal{Z} = Z(\mathrm{rs}J_Z)$, $\mathcal{S}^Z_{dc} = Z(\mathrm{rs}H^{M-M}_Z)$ and $\mathcal{Z}_{dc} = Z(\mathrm{rs}J^{M-M}_Z)$, the statement iv) holds. 

Because $\langle\Sigma\rangle = Z(\mathrm{rs}J_{Z,A})$, the statement v) holds according to Eq. (\ref{eq:rel_mea_JZA}). 

Similarly, statements iii') and iv') corresponding to Eqs. (\ref{eq:rel_bran_JX}) and (\ref{eq:rel_bran_JZ}), respectively. The statement v') corresponds to Eq. (\ref{eq:rel_bran_JZA}). 

Statements vi) and vi') correspond to Lemmas \ref{lem:dis_mea} and \ref{lem:dis_bran}, respectively. 
\end{proof}

\section{Theorem of sticking}
\label{app:sticking}

\begin{theorem}
Let $(H_X,H_Z)$ be the check matrices of the memory. For an arbitrary set of $Z$ logical operators $\Sigma$, there exists a glue code $H_G\in\mathbb{F}_2^{r_G\times n_G}$ that satisfies i) the pasting matrices $S$ and $T$ satisfy $w_{max}(S) = w_{max}(T) = 1$; 
\begin{itemize}
\item[ii)] the glue code is coarsely devised for $\Sigma$ and 
\begin{eqnarray}
n_G &=& n_N, \label{eq:nG_g} \\
r_G &\leq& w_{max}(H_X)n_N, \label{eq:rG_g} \\
w_{max}(H_G) &\leq& w_{max}(H_X); \label{eq:wmax_g}
\end{eqnarray}
\item[ii')] the glue code is finely devised for $\Sigma$ and 
\begin{eqnarray}
n_G &\leq& n_N+2(k_N-q)(q+1), \label{eq:nG_f} \\
r_G &\leq& w_{max}(H_X)n_N+2(k_N-q)(q+1), \label{eq:rG_f} \\
w_{max}(H_G) &\leq& \max\{w_{max}(H_X)+1,3\}. \label{eq:wmax_f}
\end{eqnarray}
\end{itemize}
Here, $n_N$ is the number of qubits on the support of $\Sigma$, and $k_N$ is the number of independent $Z$ logical operators on the support of $\Sigma$, and $q$ is the number of independent $Z$ logical operators in $\Sigma$. 
\label{the:sticking}
\end{theorem}

When $w_{max}(S) = 1$, we always have $\vert S\vert = 1$. 

\subsection{Graph operations}

\begin{figure}
\begin{minipage}{\linewidth}
\begin{algorithm}[H]
{\small
\begin{algorithmic}[1]
\caption{{\small BitDuplication($\mathcal{B},\mathcal{C},\mathcal{E},u,\mathcal{C}_u$)}}
\label{alg:bit}
\Statex
\State $\mathcal{B}'\leftarrow\mathcal{B}\cup\{u'\}$
\State $\mathcal{C}'\leftarrow\mathcal{C}\cup\{a'\}$
\State $\mathcal{E}'\leftarrow\mathcal{E}\cup\{(u,a'),(u',a')\}$
\For{$a\in\mathcal{C}_u$}
\State $\mathcal{E}'\leftarrow(\mathcal{E}'-\{u,a\})\cup\{(u',a)\}$
\EndFor
\State Output the Tanner graph $(\mathcal{B}',\mathcal{C}',\mathcal{E}')$. 
\end{algorithmic}
}
\end{algorithm}
\end{minipage}
\end{figure}

\begin{figure}
\begin{minipage}{\linewidth}
\begin{algorithm}[H]
{\small
\begin{algorithmic}[1]
\caption{{\small CheckDuplication($\mathcal{B},\mathcal{C},\mathcal{E},a,\mathcal{B}_a$)}}
\label{alg:check}
\Statex
\State $\mathcal{B}'\leftarrow\mathcal{B}\cup\{u'\}$
\State $\mathcal{C}'\leftarrow\mathcal{C}\cup\{a'\}$
\State $\mathcal{E}'\leftarrow\mathcal{E}\cup\{(u',a),(u',a')\}$
\For{$u\in\mathcal{B}_a$}
\State $\mathcal{E}'\leftarrow(\mathcal{E}'-\{u,a\})\cup\{(u,a')\}$
\EndFor
\State Output the Tanner graph $(\mathcal{B}',\mathcal{C}',\mathcal{E}')$. 
\end{algorithmic}
}
\end{algorithm}
\end{minipage}
\end{figure}

\begin{figure}[tbp]
\centering
\includegraphics[width=\linewidth]{\figures/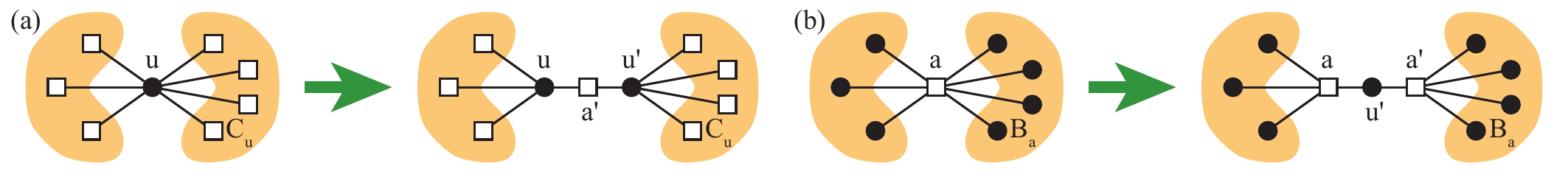}
\caption{
(a) Bit duplication. (b) Check duplication. 
}
\label{fig:duplication}
\end{figure}

To construct a finely devised glue code, we use two operations on Tanner graphs, bit duplication and check duplication. These two operations are given in Algorithms \ref{alg:bit} and \ref{alg:check}, respectively, and illustrated in Fig. \ref{fig:duplication}. In the bit duplication operation, we duplicate the bit $u$ on the Tanner graph $(\mathcal{B},\mathcal{C},\mathcal{E})$ by adding a new bit $u'$ and a check $a'$; the check $a'$ is coupled to both $u$ and $u'$; and a subset of checks that are adjacent to $u$, denoted by $\mathcal{C}_u\subseteq\mathcal{C}(\mathcal{E},u)$, are decoupled from $u$ and coupled to $u'$. Similarly, in the check duplication operation, we duplicate the check $a$ on the Tanner graph $(\mathcal{B},\mathcal{C},\mathcal{E})$ by adding a new check $a'$ and a bit $u'$; the bit $u'$ is coupled to both $a$ and $a'$; and a subset of bits that are adjacent to $a$, denoted by $\mathcal{B}_a\subseteq\mathcal{B}(\mathcal{E},a)$, are decoupled from $a$ and coupled to $a'$. These two operations have properties summarised in the following lemma. 

\begin{lemma}
Let $(\mathcal{B}',\mathcal{C}',\mathcal{E}')$ be the Tanner graph generated by applying the bit duplication or check duplication on $(\mathcal{B},\mathcal{C},\mathcal{E})$. Let $v:\mathcal{B}'\rightarrow\mathbb{F}_2$. The map $v$ is a codeword of $(\mathcal{B}',\mathcal{C}',\mathcal{E}')$ if and only if the following two conditions are satisfied, 
\begin{itemize}
\item[i)] $v$ on the domain $\mathcal{B}$ is a codeword of $(\mathcal{B},\mathcal{C},\mathcal{E})$; 
\item[ii)] For a bit duplication, $v(u') = v(u)$; 
\item[ii')] For a check duplication, $v(u') = \sum_{u\in\mathcal{B}_a}v(u)$. 
\end{itemize}
\label{lem:duplication}
\end{lemma}

\begin{proof}
{\bf Bit duplication.} The map $v$ is a codeword of $(\mathcal{B}',\mathcal{C}',\mathcal{E}')$ if and only if the following conditions are satisfied: i) $\sum_{u''\in\mathcal{B}(\mathcal{E},a)}v(u'') = 0$ for all $a\in\mathcal{C}-\mathcal{C}_u$; ii) $v(u') + \sum_{u''\in\mathcal{B}(\mathcal{E},a)-\{u\}}v(u'') = 0$ for all $a\in\mathcal{C}_u$; and iii) $v(u) = v(u')$. Under the condition iii), the condition ii) is satisfied if and only if $\sum_{u''\in\mathcal{B}(\mathcal{E},a)}v(u'') = 0$ for all $a\in\mathcal{C}_u$. Then, we can rephrase conditions i) and ii) as $\sum_{u''\in\mathcal{B}(\mathcal{E},a)}v(u'') = 0$ for all $a\in\mathcal{C}$, i.e. $v$ on the domain $\mathcal{B}$ is a codeword of $(\mathcal{B},\mathcal{C},\mathcal{E})$. 

{\bf Check duplication.} Similarly, the map $v$ is a codeword of $(\mathcal{B}',\mathcal{C}',\mathcal{E}')$ if and only if the following conditions are satisfied: i) $\sum_{u\in\mathcal{B}(\mathcal{E},a'')}v(u) = 0$ for all $a''\in\mathcal{C}-\{a\}$; ii) $v(u') + \sum_{u\in\mathcal{B}(\mathcal{E},a)-\mathcal{B}_a}v(u) = 0$; and iii) $v(u') + \sum_{u\in\mathcal{B}_a}v(u) = 0$. Under the condition iii), the condition ii) is satisfied if and only if $\sum_{u\in\mathcal{B}(\mathcal{E},a)}v(u) = 0$. Then, we can rephrase conditions i) and ii) as $\sum_{u\in\mathcal{B}(\mathcal{E},a'')}v(u) = 0$ for all $a''\in\mathcal{C}$, i.e. $v$ on the domain $\mathcal{B}$ is a codeword of $(\mathcal{B},\mathcal{C},\mathcal{E})$. 
\end{proof}

\subsection{Proof of the sicking theorem}

\begin{proof}
We prove the theorem by constructing the glue codes. 

{\bf Coarsely devised glue code - Naked glue code.} The support of $\Sigma$ is 
\begin{eqnarray}
\mathcal{B}_N = \mathcal{Q}(\Sigma).
\end{eqnarray}
Let $(\mathcal{B},\mathcal{C}_X,\mathcal{E}_X)$ be the Tanner graph of the $X$-operator check matrix $H_X$. We construct the coarsely devised glue code according to Algorithm \ref{alg:NGC}, which outputs a Tanner graph $(\mathcal{B}_N,\mathcal{C}_N,\mathcal{E}_N)$. The binary linear code of the output Tanner graph, called naked glue code, is coarsely devised for $\Sigma$. 

\begin{figure}
\begin{minipage}{\linewidth}
\begin{algorithm}[H]
{\small
\begin{algorithmic}[1]
\caption{{\small Generation of the naked glue code.}}
\label{alg:NGC}
\Statex
\State Input $(\mathcal{B},\mathcal{C}_X,\mathcal{E}_X)$ and $\mathcal{B}_N$. 
\State Find $X$-operator checks that are adjacent to bits in $\mathcal{B}_N$ on the graph $(\mathcal{B},\mathcal{C}_X,\mathcal{E}_X)$, which constitute the set of checks 
\begin{eqnarray}
\mathcal{C}_N = \bigcup_{u\in\mathcal{B}_N} \mathcal{C}_X(\mathcal{E}_X,u).
\end{eqnarray}
\State Find $X$-operator edges that are incident on bits in $\mathcal{B}_N$ on the graph $(\mathcal{B},\mathcal{C}_X,\mathcal{E}_X)$, which constitute the set of edges 
\begin{eqnarray}
\mathcal{E}_N = \bigcup_{u\in\mathcal{B}_N} \mathcal{E}_X(u).
\end{eqnarray}
\State Output the Tanner graph $(\mathcal{B}_N,\mathcal{C}_N,\mathcal{E}_N)$. 
\end{algorithmic}
}
\end{algorithm}
\end{minipage}
\end{figure}

Let $H_N$ be the check matrix of the naked glue code. Without loss of generality, we suppose that bits in $\mathcal{B}_N$ are the first $n_N = \vert\mathcal{B}_N\vert$ bits in $\mathcal{B}$, and checks in $\mathcal{C}_N$ are the first $r_N = \vert\mathcal{C}_N\vert$ checks in $\mathcal{C}_X$. Then, $H_X$ is in the form 
\begin{eqnarray}
H_X = \begin{pmatrix} H_N & A_X \\ 0_{(r_X-r_N)\times n_N} & B_X \end{pmatrix}.
\end{eqnarray}
In the general case, $H_X$ can always be transformed into the above form through permutations of rows and columns. For the naked glue code, the corresponding pasting matrices are 
\begin{eqnarray}
S_N &=& \begin{pmatrix} E_{n_N} & 0_{n_N\times(n-n_N)} \end{pmatrix}, \label{eq:SN} \\
T_N &=& \begin{pmatrix} E_{r_N} \\ 0_{(r_X-r_N)\times r_N} \end{pmatrix}.
\end{eqnarray}
Because $H_XS_N^\mathrm{T} = T_NH_N$, the naked glue code is compatible with the memory. 

Now, we prove that the naked glue code is coarsely devised for $\Sigma$. Because the support of $\Sigma$ is $\mathcal{B}_N$, $J_{Z,A}S_N^\mathrm{T}S_N = J_{Z,A}$. Using $T_N^\mathrm{T}H_XS_N^\mathrm{T} = H_N$, we have $H_NS_NJ_{Z,A}^\mathrm{T} = T_N^\mathrm{T}H_XS_N^\mathrm{T}S_NJ_{Z,A}^\mathrm{T} = T_N^\mathrm{T}H_XJ_{Z,A}^\mathrm{T} = 0$. Therefore, $(\mathrm{rs}J_{Z,A})S_N^\mathrm{T}\subseteq\mathrm{ker}H_N$, i.e. $\mathrm{rs}J_{Z,A} = (\mathrm{rs}J_{Z,A})S_N^\mathrm{T}S_N\subseteq(\mathrm{ker}H_N)S_N$. According to the Definition \ref{def:devised}, the naked glue code is coarsely devised. 

Taking $H_G = H_N$, we have $n_G = n_N = \vert\mathcal{B}_N\vert$. Because $\vert\mathcal{C}_X(\mathcal{E}_X,u)\vert\leq w_{max}(H_X)$, the number of checks $r_G = r_N = \vert\mathcal{C}_N\vert\leq w_{max}(H_X)\vert\mathcal{B}_N\vert$. The Tanner graph $(\mathcal{B}_N,\mathcal{C}_N,\mathcal{E}_N)$ is a subgraph of $(\mathcal{B},\mathcal{C}_X,\mathcal{E}_X)$, therefore, $w_{max}(H_G)=w_{max}(H_N)\leq w_{max}(H_X)$. 

{\bf Finely devised glue code - Dressed glue code.} To construct a finely devised glue code, we consider a check matrix in the form 
\begin{eqnarray}
H_D = \begin{pmatrix} H_N \\ D \end{pmatrix},
\end{eqnarray}
where the dressing matrix $D$ is taken according to Algorithm \ref{alg:dressing}. By taking pasting matrices $S_N$ and 
\begin{eqnarray}
T_D &=& \begin{pmatrix} E_{r_N} & 0_{ r_N\times (k_N-q)} \\ 0_{(r_X-r_N)\times r_N} & 0_{(r_X-r_N)\times(k_N-q)} \end{pmatrix},
\end{eqnarray}
we can find that $H_XS_N^\mathrm{T} = T_DH_D$. Therefore, such a code is always compatible with the memory. We call it dressed glue code. 

\begin{figure}
\begin{minipage}{\linewidth}
\begin{algorithm}[H]
{\small
\begin{algorithmic}[1]
\caption{{\small Generation of the dressing matrix.}}
\label{alg:dressing}
\Statex
\State Input $J_{Z,A}\in\mathbb{F}_2^{q\times n}$, $H_N\in\mathbb{F}_2^{r_N\times n_N}$ and $S_N\in\mathbb{F}_2^{n_N\times n}$. 
\State Find a basis of $(\mathrm{ker}H_N)S_N\cap(\mathrm{rs}H_Z\oplus\mathrm{rs}F_Z)$, denoted by $\{u_1,u_2,\ldots\}$. 
\State $G_0\leftarrow\begin{pmatrix} u_1 \\ u_2 \\ \vdots \end{pmatrix}S_N^\mathrm{T}$
\Comment $\mathrm{rs}G_0 \subseteq \mathrm{ker}H_N$ and $k_N = n_N-\mathrm{rank}H_N-\mathrm{rank}G_0$
\State $G_1\leftarrow J_{Z,A}S_N^\mathrm{T}$
\Comment $\mathrm{rs}G_1 \subseteq \mathrm{ker}H_N$
\State Take rows in $G_0$ and $G_1$ as basis vectors of $\mathrm{ker}H_N$ and complete the basis with vectors $\{w_1,w_2,\ldots,w_{k_N-q}\}$ that satisfy $w_jS_N \in \mathrm{rs}J_{Z,C}\oplus\mathrm{rs}H_Z\oplus\mathrm{rs}F_Z$
\State $G_2\leftarrow\begin{pmatrix} w_1 \\ w_2 \\ \vdots \\ w_{k_N-q} \end{pmatrix}$
\Comment $\begin{pmatrix} G_0 \\ G_1 \\ G_2 \end{pmatrix}$ is the generator matrix of the code $\mathrm{ker}H_N$. 
\State Find matrices $U$, $V$ and $W$ such that 
\begin{eqnarray}
G_2S_N &=& UJ_{Z,C} + VH_Z + WF_Z.
\end{eqnarray}
\Comment $U$ is always row full rank. 
\State Compute the right inverse $U^\mathrm{r}$. 
\State Output $D = {U^\mathrm{r}}^\mathrm{T}J_{X,C}S_N^T$. 
\end{algorithmic}
}
\end{algorithm}
\end{minipage}
\end{figure}

In the algorithm, we have used that $U$ is row full rank, such that its right inverse exists. Now, we prove it. If $U$ is not row full rank, there exists a nonzero vector $\alpha$ such that $\alpha U = 0$. Then, $\alpha G_2S_N = \alpha VH_Z + \alpha WF_Z \in \mathrm{rs}(G_0S_N)$. Because rows in $G_0$ and $G_2$ are linear independent, there is a contradiction. 

The dressed glue code is finely devised for $\Sigma$. In Algorithm \ref{alg:dressing}, vectors $u_j\in (\mathrm{ker}H_N)S_N$ satisfy $u_jS_N^\mathrm{T}S_N = u_j$. Then, $\mathrm{rs}(G_0S_N)\subseteq \mathrm{rs}H_Z\oplus\mathrm{rs}F_Z$ and $DG_0^\mathrm{T} = 0$. Additionally, $DG_1^\mathrm{T} = 0$ and $DG_2^\mathrm{T} = E_{k_N-q}$. Here, we have used that $J_{Z,A}S_N^\mathrm{T}S_N = J_{Z,A}$. Therefore, the dressing matrix $D$ removes basis vectors of $\mathrm{rs}\bar{G}_2$ from the basis of $\mathrm{ker}H_D$, i.e. the generator matrix of the dressed glue code is 
\begin{eqnarray}
\begin{pmatrix} G_0 \\ G_1 \end{pmatrix},
\end{eqnarray}
Noticing the definitions of $G_0$ and $G_1$ in Algorithm \ref{alg:dressing}, we have proved that the dressed glue code is finely devised. 

{\bf Finely devised LDPC glue code.} The dressed glue code may not satisfy the LDPC condition. Now, we generate an LDPC glue code from the dressed glue code, which is finely devised. 

\begin{figure}
\begin{minipage}{\linewidth}
\begin{algorithm}[H]
{\small
\begin{algorithmic}[1]
\caption{{\small Generation of the finely devised LDPC glue code.}}
\label{alg:LDPC}
\Statex
\State Input $\mathcal{B}_N$, $\mathcal{C}_N$, $\mathcal{E}_N$, $\mathcal{C}_D$, $\mathcal{E}_D$. 
\State $\mathcal{B}_G\leftarrow\mathcal{B}_N$
\State $\mathcal{C}_G\leftarrow\mathcal{C}_N\cup\mathcal{C}_D$
\State $\mathcal{E}_G\leftarrow\mathcal{E}_N\cup\mathcal{E}_D$
\For{$u\in\mathcal{B}_N$}
\While{$\vert\mathcal{C}_G(\mathcal{E}_G,u)-\mathcal{C}_N(\mathcal{E}_N,u)\vert>1$}
\State Choose $a,a'\in\mathcal{C}_G(\mathcal{E}_G,u)-\mathcal{C}_N(\mathcal{E}_N,u)$. 
\State $(\mathcal{B}_G,\mathcal{C}_G,\mathcal{E}_G)\leftarrow$ BitDuplication($\mathcal{B}_G,\mathcal{C}_G,\mathcal{E}_G,u,\{a,a'\})$. 
\EndWhile
\EndFor
\For{$a\in\mathcal{C}_N$}
\While{$\vert\mathcal{B}_G(\mathcal{E}_G,a)-\mathcal{B}_N(\mathcal{E}_N,a)\vert>1$}
\State Choose $u,u'\in\mathcal{B}_G(\mathcal{E}_G,a)-\mathcal{B}_N(\mathcal{E}_N,a)$. 
\State $(\mathcal{B}_G,\mathcal{C}_G,\mathcal{E}_G)\leftarrow$ CheckDuplication($\mathcal{B}_G,\mathcal{C}_G,\mathcal{E}_G,a,\{u,u'\})$. 
\EndWhile
\EndFor
\State Output the Tanner graph $(\mathcal{B}_G,\mathcal{C}_G,\mathcal{E}_G)$. 
\end{algorithmic}
}
\end{algorithm}
\end{minipage}
\end{figure}

Let $(\mathcal{B}_N,\mathcal{C}_N,\mathcal{E}_N)$ and $(\mathcal{B}_N,\mathcal{C}_D,\mathcal{E}_D)$ be Tanner graphs of the naked glue code and dressing matrix $D$, respectively. Then the Tanner graph of the dressed glue code is $(\mathcal{B}_N,\mathcal{C}_N\cup\mathcal{C}_D,\mathcal{E}_N\cup\mathcal{E}_D)$. We generate the LDPC glue code by applying the bit duplication and check duplication operations on the Tanner graph according to Algorithm \ref{alg:LDPC}. On the generated Tanner graph $(\mathcal{B}_G,\mathcal{C}_G,\mathcal{E}_G)$, the vertex degrees of bits $u\in\mathcal{B}_N$ (checks $a\in\mathcal{C}_N$) are not larger than $w_{max}(H_N)+1$, the vertex degrees of bits (checks) added in bit (check) duplication operations are three, and the vertex degrees of bits (checks) added in check (bit) duplication operations are two. Let $H_G$ be the check matrix of $(\mathcal{B}_G,\mathcal{C}_G,\mathcal{E}_G)$. Then, $w_{max}(H_G)\leq \max\{w_{max}(H_N)+1,3\}$. 

Now, we prove that the LDPC glue code $(\mathcal{B}_G,\mathcal{C}_G,\mathcal{E}_G)$ is a finely devised for $\Sigma$. Its check matrix is in the form 
\begin{eqnarray}
H_G = \begin{pmatrix} H_N & 0 \\ A_G & B_G \end{pmatrix}.
\end{eqnarray}
By taking pasting matrices 
\begin{eqnarray}
S &=& \begin{pmatrix} E_{n_N} & 0_{n_N\times(n-n_N)} \\ 0 & 0 \end{pmatrix}, \\
T &=& \begin{pmatrix} E_{r_N} & 0 \\ 0_{(r_X-r_N)\times r_N} & 0 \end{pmatrix},
\end{eqnarray}
we can find that $H_XS^\mathrm{T} = TH_G$. Therefore, the code is always compatible with the memory. According to Lemma \ref{lem:duplication}, $(\mathrm{ker}H_G)S = (\mathrm{ker}H_D)S_N$. In the proof for the dressed glue code, we have proved that $\mathrm{rs}J_{Z,A} = (\mathrm{ker}H_D)S_N$. Then, $\mathrm{rs}J_{Z,A} = (\mathrm{ker}H_G)S$, i.e. the LDPC glue code is finely devised. 

In duplication operations, the number of bits and checks added to the Tanner graph depends on how we choose $G_2$ (i.e. $U$). To minimise the number of bits and checks, we use the standard form of various generator matrices. First, there always exist an invertible matrix $R$ and a permutation matrix $\pi_1$ such that $J_X = RJ_X^S\pi_1$ and $J_Z = R^{-1}J_Z^S\pi_1$, where $J_X^S = \begin{pmatrix} E_k & 0 & J'_X \end{pmatrix}$ and $J_Z^S = \begin{pmatrix} E_k & J'_Z & 0 \end{pmatrix}$. Accordingly, the matrix $S_N$ is always in the form $S_N = \begin{pmatrix} S_K & S_Z & 0 \end{pmatrix}\pi_1$. This expression of $S_N$ is consistent with Eq. (\ref{eq:SN}) up to the column permutation $\pi_1$, and $w_{max}(S_K) = w_{max}(S_Z) = 1$. Second, the $Z$ logical operators to be measured are $J_{Z,A} = (\bar{J}_Z)_{1:q,\bullet}J_Z$, where $(\bar{J}_Z)_{1:q,\bullet}$ denotes the first $q$ rows of $\bar{J}_Z$. Without loss of generality, we can always choose $(\bar{J}_Z)_{1:q,\bullet}$ such that$(\bar{J}_Z)_{1:q,\bullet} = \begin{pmatrix} E_q & P \end{pmatrix}\pi_2R$, where $\pi_2$ is a permutation matrix. Accordingly, the $Z$ logical operators to be measured are $J_{Z,A} = \begin{pmatrix} E_q & P \end{pmatrix}\pi_2J_Z^S\pi_1$, and the $X$ logical operators preserved in the measurement are $J_{X,C} = \begin{pmatrix} P^\mathrm{T} & E_{k-q} \end{pmatrix}\pi_2J_X^S\pi_1$. Then, $D = {U^\mathrm{r}}^\mathrm{T}\begin{pmatrix} P^\mathrm{T} & E_{k-q} \end{pmatrix}\pi_2 S_K^T$. Third, we can always choose $w_j$ vectors such that the matrix $U$ is in the form $U = \begin{pmatrix} E_{k_N-q} & Q \end{pmatrix}\pi_3$, where $\pi_3$ is a permutation matrix. Then, ${U^\mathrm{r}}^\mathrm{T} = \begin{pmatrix} E_{k_N-q} & 0 \end{pmatrix}\pi_3$, and the number of nonzero entries in ${U^\mathrm{r}}^\mathrm{T}\begin{pmatrix} P^\mathrm{T} & E_{k-q} \end{pmatrix}$ is not larger than $(k_N-q)(q+1)$. Therefore, the number of nonzero entries in $D$ is not larger than $(k_N-q)(q+1)$. 

In bit duplication operations, the number of bits (checks) added to the Tanner graph is 
\begin{eqnarray}
\sum_{u\in\mathcal{B}_N} \max\{0,\vert\mathcal{C}_D(\mathcal{E}_D,u)\vert-1\} \leq (k_N-q)(q+1).
\end{eqnarray}
In check duplication operations, the number of bits (checks) added to the Tanner graph is 
\begin{eqnarray}
\sum_{a\in\mathcal{C}_D} \max\{0,\vert\mathcal{B}_N(\mathcal{E}_D,a)\vert-1\} \leq (k_N-q)(q+1).
\end{eqnarray}
\end{proof}

\section{Simultaneous measurement with a logical thickness}
\label{app:thickness}

We consider logical-operator generator matrices in the standard form, i.e. $J_X = \begin{pmatrix} E_k & 0 & J'_X \end{pmatrix}$ and $J_Z = \begin{pmatrix} E_k & J'_Z & 0 \end{pmatrix}$ ($\pi_1$ and $R$ are identity matrices). If the simultaneous measurement has a logical thickness of $t$, the matrix $(\bar{J}_Z)_{1:q,\bullet}$ is in the form 
\begin{eqnarray}
(\bar{J}_Z)_{1:q,\bullet} &=& \begin{pmatrix}
\bar{J}_{Z,1} & 0 & \cdots \\
0 & \bar{J}_{Z,2} & \cdots \\
\vdots & \vdots & \ddots
\end{pmatrix}\pi_2' = \left(\bigoplus_l \bar{J}_{Z,l}\right)\pi_2',
\end{eqnarray}
where each block $\bar{J}_{Z,l}$ corresponds to a subset of logical operators $\Sigma_l$ (see the definition of logical thickness in the main text), $\pi_2'$ is a permutation matrix, $\bar{J}_{Z,l} \in \mathbb{F}_2^{a_l\times b_l}$ and $a_l\leq t$ for all $l$. Here, $t$ is the logical thickness of the operator set $\Sigma$. Without loss of generality, each block is in the standard from $\bar{J}_{Z,l} = \begin{pmatrix} E_{a_l} & P_l \end{pmatrix}$. Up to a permutation of columns, $(\bar{J}_Z)_{1:q,\bullet} = \begin{pmatrix} E_q & P \end{pmatrix}\pi_2$ (notice that $\pi_2$ is different from $\pi_2'$), where $P = \bigoplus_l P_l$. Then, the number of nonzero entries in ${U^\mathrm{r}}^\mathrm{T}\begin{pmatrix} P^\mathrm{T} & E_{k-q} \end{pmatrix}$ is not larger than $(k_N-q)(t+1)$, and number of nonzero entries in $D$ is also not larger than $(k_N-q)(t+1)$. Accordingly, the finely devised glue code satisfies 
\begin{eqnarray}
n_G &\leq& n_N+2(k_N-q)(t+1), \\
r_G &\leq& w_{max}(H_X)n_N+2(k_N-q)(t+1).
\end{eqnarray}

\section{General logical Pauli measurements and universal quantum computing}
\label{app:general}

In addition to $Z$ logical operators, we can also measure $X$ logical operators in a similar way. Suppose there are ancilla logical qubits encoded in a block independent from the memory, we can also measure logical operators in the from $\sigma_M\sigma_A$, where $\sigma_M$ ($\sigma_A$) is an $X$ or $Z$ logical operator of the memory (ancilla block). Notice that $\sigma_M$ and $\sigma_A$ could be different in the $X/Z$ species. 

\begin{figure}[tbp]
\centering
\includegraphics[width=\linewidth]{\figures/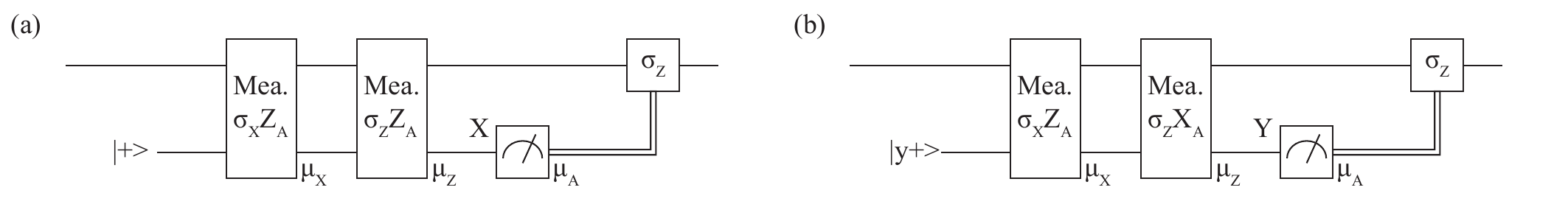}
\caption{
Circuits for measurements of general logical Pauli operators. 
}
\label{fig:circuit}
\end{figure}

For a general logical Pauli operator $\sigma$, we can measure it in the following way. We write the Pauli operator in the form $\sigma = \nu\sigma_X\sigma_Z$, where $\nu = \pm 1,\pm i$, and $\sigma_X$ ($\sigma_Z$) is an $X$ ($Z$) logical operator. Let $X_A,Y_A,Z_A$ be Pauli operators of an ancilla logical qubit, respectively. We can measure $\sigma$ using the ancilla logical qubit: 
\begin{itemize}
\item[$\bullet$] If $[\sigma_X,\sigma_Z] = 0$, we measure $\sigma$ according to Fig. \ref{fig:circuit}(a): first, we initialise the ancilla logical qubit in the state $\ket{+}$; then, we measure $\sigma_XZ_A$, $\sigma_ZZ_A$ and $X_A$. Let $\mu_X,\mu_Z,\mu_A$ be outcomes of three measurements, respectively. The measurement outcome of $\sigma$ is $\nu\mu_X\mu_Z$. When $\mu_A = -1$, we apply a correction gate $\sigma_Z$. 
\item[$\bullet$] If $\{\sigma_X,\sigma_Z\} = 0$, we measure $\sigma$ according to Fig. \ref{fig:circuit}(b): first, we initialise the ancilla logical qubit in the state $\ket{y+}$; then, we measure $\sigma_XZ_A$, $\sigma_ZX_A$ and $Y_A$. Let $\mu_X,\mu_Z,\mu_A$ be outcomes of three measurements, respectively. The measurement outcome of $\sigma$ is $-i\nu\mu_X\mu_Z$. When $\mu_A = -1$, we apply a correction gate $\sigma_Z$. 
\end{itemize}

Next, we generalise the above protocol for the simultaneous measurement of general logical Pauli operators. 

\subsection{Simultaneous measurement of general logical Pauli operators}

Let $\Theta = \{\sigma_1 = \nu_1\sigma_{X,1}\sigma_{Z,1},\sigma_2 = \nu_2\sigma_{X,2}\sigma_{Z,2},\ldots,\sigma_q = \nu_q\sigma_{X,q}\sigma_{Z,q}\}$ be an arbitrary set of commutative logical Pauli operators, where $\nu_j = \pm 1,\pm i$, and $\{\sigma_{X,j}\}$ ($\{\sigma_{Z,j}\}$) are $X$ ($Z$) logical operators of the memory. 

\begin{definition}
{\bf Characteristic number.} The number $\eta$ characterises the commutation relation between two sub-operators: $\eta_j = 0$ if and only if $[\sigma_{X,j},\sigma_{Z,j}] = 0$, and $\eta_j = 1$ if and only if $\{\sigma_{X,j},\sigma_{Z,j}\} = 0$. When $\eta_j = 0$, $\nu_j = \pm 1$; when $\eta_j = 1$, $\nu_j = \pm i$. 
\end{definition}

To measure $\Theta$, we need two independent ancilla blocks, $A0$ and $A1$. We suppose that each block encodes at least $q$ logical qubits, though not all of them are necessarily used. We use $X_{\alpha,j},Y_{\alpha,j},Z_{\alpha,j}$ to denote Pauli operators of the $j$th logical qubit in the block $A\alpha$. 

The protocol for the simultaneous measurement of general Pauli operators has the following steps: 
\begin{itemize}
\item[1.] Initialise ancilla logical qubits in blocks $A0$ and $A1$ in states $\ket{+}$ and $\ket{y+}$, respectively; 
\item[2.] Simultaneously measure logical operators $\Omega_X = \{\sigma_{X,j}Z_{0,j}^{1-\eta_j}Z_{1,j}^{\eta_j}\}$ using the devised sticking protocol or brute-force branching protocol; 
\item[3.] Simultaneously measure logical operators $\Omega_Z = \{\sigma_{Z,j}Z_{0,j}^{1-\eta_j}X_{1,j}^{\eta_j}\}$ using the devised sticking protocol or brute-force branching protocol; 
\item[4.] Measure ancilla logical qubits in blocks $A0$ and $A1$ in bases $X$ and $Y$, respectively; 
\item[5.] Apply correction gates according to measurement outcomes. 
\end{itemize}

For a quantum LDPC code, we can initialise logical qubits in a block in the state $\ket{+}$ with a time cost independent of the logical qubit number in the block, and the time cost for the measurement in the basis $X$ is also independent of the logical qubit number in the block. If we choose a code with transversal $S$ gate for the block $A1$, the initialisation in the state $\ket{y+}$ and measurement in the basis $Y$ can also be accomplished in time independent of the logical qubit number in the block; see Sec. \ref{app:universal} for a discussion on the general case. Using devised sticking or brute-force branching to implement the measurements on $\Omega_X$ and $\Omega_Z$, the time cost is independent of sizes of $\Omega_X$ and $\Omega_Z$. Overall, we can measure the operator set $\Theta$ in time independent of the size $q$ of the operator set. 

\begin{definition}
{\bf Regular operator set.} The operator set $\Theta$ is said to be regular if and only if $[\sigma_{X,i},\sigma_{Z,j}] = 0$ for all $i\neq j$. 
\end{definition}

\begin{lemma}
The above protocol realises the measurement on the operator set $\Theta$ if $\Theta$ is regular. 
\label{lem:general}
\end{lemma}

\begin{proof}
Let $\ket{\psi}$ be the logical state of the memory. The overall state after step-1 is 
\begin{eqnarray}
\ket{\Psi_i} = \ket{\psi}\otimes\left(\bigotimes_{j=1}^q\ket{+}_{0,j}\right)\otimes\left(\bigotimes_{j=1}^q\ket{y+}_{1,j}\right).
\end{eqnarray}
After steps 2, 3, 4 and 5, the state is 
\begin{eqnarray}
\ket{\Psi_f} &=& \left(\prod_{j=1}^q\sigma_{Z,j}^{\delta_{\mu_{0,j},-1}^{1-\eta_j}\delta_{\mu_{1,j},-1}^{\eta_j}}\right) \notag \\
&&\times \left(\prod_{j=1}^q\frac{\openone+\mu_{1,j}Y_{1,j}}{2}\right)\left(\prod_{j=1}^q\frac{\openone+\mu_{0,j}X_{0,j}}{2}\right) \notag \\
&&\times \left(\prod_{j=1}^q\frac{\openone+\mu_{Z,j}\sigma_{Z,j}Z_{0,j}^{1-\eta_j}X_{1,j}^{\eta_j}}{2}\right)\left(\prod_{j=1}^q\frac{\openone+\mu_{X,j}\sigma_{X,j}Z_{0,j}^{1-\eta_j}Z_{1,j}^{\eta_j}}{2}\right) \ket{\Psi_i},
\end{eqnarray}
where $\mu_{X,j},\mu_{Z,j},\mu_{0,j},\mu_{1,j} = \pm 1$ are corresponding measurement outcomes. The first line describes the correction gates applied in step 5, noticing that $\delta_{\mu_{0,j},-1}^{1-\eta_j}\delta_{\mu_{1,j},-1}^{\eta_j} = \delta_{\mu_{0,j},-1}$ ($\delta_{\mu_{0,j},-1}^{1-\eta_j}\delta_{\mu_{1,j},-1}^{\eta_j} = \delta_{\mu_{1,j},-1}$) when $\eta_j = 0$ ($\eta_j = 1$). The second and third lines describe measurements applied in steps 2, 3 and 4. Here, $(\openone+\mu\sigma)/2$ is the projection operator describing the measurement of Pauli operator $\sigma$ with the outcome $\mu$. 

Let's consider the case that $q = 1$. The state becomes 
\begin{eqnarray}
\ket{\Psi_f} &=& \sigma_{Z}^{\delta_{\mu_{0},-1}^{1-\eta}\delta_{\mu_{1},-1}^{\eta}}
\frac{\openone+\mu_{1}Y_{1}}{2} \frac{\openone+\mu_{0}X_{0}}{2}
\frac{\openone+\mu_{Z}\sigma_{Z}Z_{0}^{1-\eta}X_{1}^{\eta}}{2} \frac{\openone+\mu_{X}\sigma_{X}Z_{0}^{1-\eta}Z_{1}^{\eta}}{2} \ket{\Psi_i},
\end{eqnarray}
where we have neglected the subscript $j = 1$ for simplicity. Because ancilla logical qubits are initialised in states $\ket{+}$ and $\ket{y+}$, 
\begin{eqnarray}
\ket{\Psi_f} &=& \sigma_{Z}^{\delta_{\mu_{0},-1}^{1-\eta}\delta_{\mu_{1},-1}^{\eta}}
\frac{\openone+\mu_{1}Y_{1}}{2} \frac{\openone+\mu_{0}X_{0}}{2}
\frac{\openone+\mu_{Z}\sigma_{Z}Z_{0}^{1-\eta}X_{1}^{\eta}}{2} \frac{\openone+\mu_{X}\sigma_{X}Z_{0}^{1-\eta}Z_{1}^{\eta}}{2}
\frac{\openone+Y_{1}}{2} \frac{\openone+X_{0}}{2} \ket{\Psi_i}.
\end{eqnarray}
Using $(\openone+\sigma)(\openone+\tau) = (\openone+\sigma)(\openone+\sigma\tau)$, we can rewrite the state as 
\begin{eqnarray}
\ket{\Psi_f} &=& \sigma_{Z}^{\delta_{\mu_{0},-1}^{1-\eta}\delta_{\mu_{1},-1}^{\eta}}
\frac{\openone+\mu_{1}Y_{1}}{2} \frac{\openone+\mu_{0}X_{0}}{2}
\frac{\openone+\mu_{Z}\sigma_{Z}Z_{0}^{1-\eta}X_{1}^{\eta}}{2} \frac{\openone+(-i)^{\eta}\nu\mu_{X}\mu_{Z}\sigma Y_{1}^{\eta}}{2}
\frac{\openone+Y_{1}}{2} \frac{\openone+X_{0}}{2} \ket{\Psi_i}.
\end{eqnarray}
Using $\sigma(\openone+\sigma)/2 = (\openone+\sigma)/2$, we can further rewrite the state as  
\begin{eqnarray}
\ket{\Psi_f} &=& \sigma_{Z}^{\delta_{\mu_{0},-1}^{1-\eta}\delta_{\mu_{1},-1}^{\eta}}
\frac{\openone+\mu_{1}Y_{1}}{2} \frac{\openone+\mu_{0}X_{0}}{2}
\frac{\openone+\mu_{Z}\sigma_{Z}Z_{0}^{1-\eta}X_{1}^{\eta}}{2} \frac{\openone+(-i)^{\eta}\nu\mu_{X}\mu_{Z}\sigma}{2}
\frac{\openone+Y_{1}}{2} \frac{\openone+X_{0}}{2} \ket{\Psi_i},
\end{eqnarray}
in which we have removed $Y_{1}^{\eta}$. Because of the commutativity of operators, 
\begin{eqnarray}
\ket{\Psi_f} &=& \sigma_{Z}^{\delta_{\mu_{0},-1}^{1-\eta}\delta_{\mu_{1},-1}^{\eta}}
\frac{\openone+\mu_{1}Y_{1}}{2} \frac{\openone+\mu_{0}X_{0}}{2}
\frac{\openone+\mu_{Z}\sigma_{Z}Z_{0}^{1-\eta}X_{1}^{\eta}}{2}
\frac{\openone+Y_{1}}{2} \frac{\openone+X_{0}}{2} \notag \\
&&\times \frac{\openone+(-i)^{\eta}\nu\mu_{X}\mu_{Z}\sigma}{2} \ket{\Psi_i}.
\end{eqnarray}
For two Pauli operators $\sigma$ and $\tau$, 
\begin{eqnarray}
\frac{\openone+\mu\sigma}{2}\tau\frac{\openone+\sigma}{2} = \delta_{\mu,+1}\frac{\openone+\mu\sigma}{2}\tau\frac{\openone+\sigma}{2}
\end{eqnarray}
if $[\sigma,\tau] = 0$, and 
\begin{eqnarray}
\frac{\openone+\mu\sigma}{2}\tau\frac{\openone+\sigma}{2} = \delta_{\mu,-1}\frac{\openone+\mu\sigma}{2}\tau\frac{\openone+\sigma}{2}
\end{eqnarray}
if $\{\sigma,\tau\} = 0$. Then, 
\begin{eqnarray}
\ket{\Psi_f} &=& \sigma_{Z}^{\delta_{\mu_{0},-1}^{1-\eta}\delta_{\mu_{1},-1}^{\eta}}
\frac{\openone+\mu_{1}Y_{1}}{2} \frac{\openone+\mu_{0}X_{0}}{2}
\frac{\delta_{\mu_{0},+1}^{1-\eta}\delta_{\mu_{1},+1}^{\eta}\openone+\delta_{\mu_{0},-1}^{1-\eta}\delta_{\mu_{1},-1}^{\eta}\mu_{Z}\sigma_{Z}Z_{0}^{1-\eta}X_{1}^{\eta}}{2}
\frac{\openone+Y_{1}}{2} \frac{\openone+X_{0}}{2} \notag \\
&&\times \frac{\openone+(-i)^{\eta}\nu\mu_{X}\mu_{Z}\sigma}{2} \ket{\Psi_i} \notag \\
&=& \Delta \frac{\openone+(-i)^{\eta}\nu\mu_{X}\mu_{Z}\sigma}{2} \ket{\Psi_i}.
\end{eqnarray}
Notice that the operator 
\begin{eqnarray}
\Delta = \frac{\openone+\mu_{1}Y_{1}}{2} \frac{\openone+\mu_{0}X_{0}}{2}
\frac{\delta_{\mu_{0},+1}^{1-\eta}\delta_{\mu_{1},+1}^{\eta}\openone+\delta_{\mu_{0},-1}^{1-\eta}\delta_{\mu_{1},-1}^{\eta}\mu_{Z}Z_{0}^{1-\eta}X_{1}^{\eta}}{2}
\frac{\openone+Y_{1}}{2} \frac{\openone+X_{0}}{2}
\end{eqnarray}
only acts on ancilla logical qubits. 

For a general $q$, because ancilla logical qubits are initialised in states $\ket{+}$ and $\ket{y+}$, 
\begin{eqnarray}
\ket{\Psi_f} &=& \left(\prod_{j=1}^q\sigma_{Z,j}^{\delta_{\mu_{0,j},-1}^{1-\eta_j}\delta_{\mu_{1,j},-1}^{\eta_j}}\right) \notag \\
&&\times \left(\prod_{j=1}^q\frac{\openone+\mu_{1,j}Y_{1,j}}{2}\right)\left(\prod_{j=1}^q\frac{\openone+\mu_{0,j}X_{0,j}}{2}\right) \notag \\
&&\times \left(\prod_{j=1}^q\frac{\openone+\mu_{Z,j}\sigma_{Z,j}Z_{0,j}^{1-\eta_j}X_{1,j}^{\eta_j}}{2}\right)\left(\prod_{j=1}^q\frac{\openone+\mu_{X,j}\sigma_{X,j}Z_{0,j}^{1-\eta_j}Z_{1,j}^{\eta_j}}{2}\right) \notag \\
&&\times \left(\prod_{j=1}^q\frac{\openone+Y_{1,j}}{2}\right)\left(\prod_{j=1}^q\frac{\openone+X_{0,j}}{2}\right) \ket{\Psi_i}.
\end{eqnarray}
Now, it is crucial to use the condition that $[\sigma_{X,i},\sigma_{Z,j}] = 0$ for all $i\neq j$. Under the condition, 
\begin{eqnarray}
\ket{\Psi_f} &=& \prod_{j=1}^q\left(\sigma_{Z,j}^{\delta_{\mu_{0,j},-1}^{1-\eta_j}\delta_{\mu_{1,j},-1}^{\eta_j}}\right. \notag \\
&&\times \frac{\openone+\mu_{1,j}Y_{1,j}}{2} \frac{\openone+\mu_{0,j}X_{0,j}}{2} \notag \\
&&\times \frac{\openone+\mu_{Z,j}\sigma_{Z,j}Z_{0,j}^{1-\eta_j}X_{1,j}^{\eta_j}}{2} \frac{\openone+\mu_{X,j}\sigma_{X,j}Z_{0,j}^{1-\eta_j}Z_{1,j}^{\eta_j}}{2} \notag \\
&&\times \left.\frac{\openone+Y_{1,j}}{2} \frac{\openone+X_{0,j}}{2}\right) \ket{\Psi_i}.
\end{eqnarray}
Similar to the case of $q = 1$, we have 
\begin{eqnarray}
\ket{\Psi_f} &=& \prod_{j=1}^q\left(\Delta_j\frac{\openone+(-i)^{\eta_j}\nu_j\mu_{X,j}\mu_{Z,j}\sigma_j}{2}\right) \ket{\Psi_i} = \left(\prod_{j=1}^q\frac{\openone+(-i)^{\eta_j}\nu_j\mu_{X,j}\mu_{Z,j}\sigma_j}{2}\right) \ket{\psi}\otimes\ket{\phi}_{A1,A2}
\end{eqnarray}
where operators $\{\Delta_j\}$ only act on ancilla logical qubits, and $\ket{\phi}_{A1,A2}$ is a state of ancilla logical qubits. Notice that 
\begin{eqnarray}
\braket{\phi}{\phi}_{A1,A2} &=& \prod_{j=1}^q \bra{+}_{0,j}\otimes\bra{y+}_{1,j}\Delta_j\ket{+}_{0,j}\otimes\ket{y+}_{1,j}
\end{eqnarray}
is independent of measurement outcomes $\{\mu_{X,j},\mu_{Z,j}\}$: $\Delta_j$ depends on $\mu_{Z,j}$, however, $\mu_{Z,j}$ always appears as an overall factor in the state $\ket{\phi}_{A1,A2}$, which does not change the norm of the state. Therefore, the protocol realises the measurement of $\{\sigma_j\}$, and measurement outcomes are $\{(-i)^{\eta_j}\nu_j\mu_{X,j}\mu_{Z,j}\}$, respectively. 
\end{proof}

\subsection{Regularisation of the generating set}

According to Lemma \ref{lem:general}, the protocol only works for a regular operator set. For an arbitrary set $\Theta$ of commutative Pauli operators, we need to regularise the operator set before measurement: by measuring operators $\Theta$, we actually measure all operators in the group $\langle\Theta\rangle$; therefore, measuring any generating set of the group is equivalent to measuring $\Theta$. We can always find two regular subsets $\Theta'$ and $\Theta''$ such that $\Theta'\cup\Theta''$ is the generating set of the group. By measuring $\Theta'$ and $\Theta''$, we can realise the measurement of $\Theta$. 

\begin{figure}
\begin{minipage}{\linewidth}
\begin{algorithm}[H]
{\small
\begin{algorithmic}[1]
\caption{{\small Regularisation of the generating set.}}
\label{alg:regularisation}
\Statex
\State Input $\Theta$. 
\State $\Theta_0\leftarrow\Theta$
\State $\Theta_1,\Theta_2,\Theta_3,\Theta_4\leftarrow\emptyset$
\While{$\Theta_0\neq \emptyset$}
\If {$\exists\sigma = \nu\sigma_X\sigma_Z\in\Theta_0$ such that $\{\sigma_X,\sigma_Z\} = 0$}
\State $\Theta_1\leftarrow\Theta_1\cup\{\sigma\}$
\State $\Theta_0\leftarrow\Theta_0-\{\sigma\}$
\For {$\tau = \nu_\tau\tau_X\tau_Z\in\Theta_0$}
\If {$\{\sigma_X,\tau_Z\} = 0$}
\State $\Theta_0\leftarrow\Theta_0-\{\tau\}$
\State $\Theta_0\leftarrow\Theta_0\cup\{\sigma\tau\}$
\EndIf
\EndFor
\Else
\State Choose an arbitrary element $\sigma = \nu\sigma_X\sigma_Z\in\Theta_0$. 
\If {$\exists\sigma' = \nu'\sigma'_X\sigma'_Z\in\Theta_0$ such that $\{\sigma_X,\sigma'_Z\} = 0$}
\State $\Theta_3\leftarrow\Theta_3\cup\{\sigma\}$
\State $\Theta_4\leftarrow\Theta_4\cup\{\sigma'\}$
\State $\Theta_0\leftarrow\Theta_0-\{\sigma,\sigma'\}$
\For {$\tau = \nu_\tau\tau_X\tau_Z\in\Theta_0$}
\If {$\{\sigma_X,\tau_Z\} = \{\sigma'_X,\tau_Z\} = 0$}
\State $\Theta_0\leftarrow\Theta_0-\{\tau\}$
\State $\Theta_0\leftarrow\Theta_0\cup\{\sigma\sigma'\tau\}$
\EndIf
\If {$\{\sigma_X,\tau_Z\} = [\sigma'_X,\tau_Z] = 0$}
\State $\Theta_0\leftarrow\Theta_0-\{\tau\}$
\State $\Theta_0\leftarrow\Theta_0\cup\{\sigma'\tau\}$
\EndIf
\If {$[\sigma_X,\tau_Z] = \{\sigma'_X,\tau_Z\} = 0$}
\State $\Theta_0\leftarrow\Theta_0-\{\tau\}$
\State $\Theta_0\leftarrow\Theta_0\cup\{\sigma\tau\}$
\EndIf
\EndFor
\Else
\State $\Theta_2\leftarrow\Theta_2\cup\{\sigma\}$
\State $\Theta_0\leftarrow\Theta_0-\{\sigma\}$
\EndIf
\EndIf
\EndWhile
\State Output $\Theta' = \Theta_1\cup\Theta_2\cup\Theta_3$ and $\Theta'' = \Theta_4$. 
\end{algorithmic}
}
\end{algorithm}
\end{minipage}
\end{figure}

\begin{lemma}
For an arbitrary set $\Theta$ of commutative Pauli operators, there exists a generating set $\Theta'\cup\Theta''$ of the group $\langle\Theta\rangle$ such that two subsets $\Theta'$ and $\Theta''$ are regular. 
\label{lem:regularisation}
\end{lemma}

\begin{proof}
We can prove the lemma by constructing $\Theta'$ and $\Theta''$ according to Algorithm \ref{alg:regularisation}. First, $\Theta_1\cup\Theta_2\cup\Theta_3\cup\Theta_4$ is a generating set of $\langle\Theta\rangle$. 

Second, after each round of the while loop, the following statements holds: 
\begin{itemize}
\item[i)] For all $\sigma = \nu\sigma_X\sigma_Z\in\Theta_1$ and $\tau = \nu_\tau\tau_X\tau_Z\in\Theta_0\cup\Theta_1\cup\Theta_2\cup\Theta_3\cup\Theta_4-\{\sigma\}$, $[\sigma_X,\tau_Z] = [\sigma_Z,\tau_X] = 0$; 
\item[ii)] For all $\sigma = \nu\sigma_X\sigma_Z\in\Theta_2$ and $\tau = \nu_\tau\tau_X\tau_Z\in\Theta_0\cup\Theta_1\cup\Theta_2\cup\Theta_3\cup\Theta_4-\{\sigma\}$, $[\sigma_X,\tau_Z] = [\sigma_Z,\tau_X] = 0$; 
\item[iii)] For all $\sigma = \nu\sigma_X\sigma_Z\in\Theta_3$ and $\tau = \nu_\tau\tau_X\tau_Z\in\Theta_0\cup\Theta_1\cup\Theta_2\cup\Theta_3-\{\sigma\}$, $[\sigma_X,\tau_Z] = [\sigma_Z,\tau_X] = 0$; 
\item[iv)] For all $\sigma = \nu\sigma_X\sigma_Z\in\Theta_4$ and $\tau = \nu_\tau\tau_X\tau_Z\in\Theta_0\cup\Theta_1\cup\Theta_2\cup\Theta_4-\{\sigma\}$, $[\sigma_X,\tau_Z] = [\sigma_Z,\tau_X] = 0$. 
\end{itemize}
In each round of the while loop, we move one or two operators from $\Theta_0$ to other subsets. If there exists an operator $\sigma\in\Theta_0$ with two anti-commutative sub-operators (line 5), we move it to $\Theta_1$. By carrying out lines 8 to 11, we make sure that the statement i) holds. If such an operator does not exist, there are two cases. In the first case, there exists $\sigma'\in\Theta_0$ satisfy the condition in line 14. Then, we move $\sigma$ and $\sigma'$ to $\Theta_3$ and $\Theta_4$, respectively. By carrying out lines 19 to 27, we make sure that statements iii) and iv) hold; notice that two sub-operators of $\sigma$ ($\sigma'$) are commutative. In the second case, $\sigma'$ does not exist. Then, we move $\sigma$ to $\Theta_2$. The statement ii) holds. 
\end{proof}

\subsection{Universal quantum computing and magic state duplication}
\label{app:universal}

With measurements on $X$ and $Z$ logical operators, we can realise the logical controlled-NOT gate \cite{Horsman2012}; using measurements with $[\sigma_X,\sigma_Z] = 0$, we can realise the logical Hadamard gate: to apply the gate on logical qubit-$1$, we prepare logical qubit-$2$ in the state $\ket{0}$, apply the measurement $\bar{Z}_1\bar{X}_2$ on two logical qubits and measure logical qubit-$1$ in the $X$ basis; these operations transfer the state of logical qubit-$1$ to logical qubit-$2$ with the basis rotated. With the logical controlled-NOT gate and Hadamard gate, we can distill $\ket{y+}$ magic states and realise the logical $S$ gate; with these logical Clifford gates, we can distill the magic state for implementing the logical $T$ gate \cite{Fowler2012}. These logical gates constitute a universal gate set. 

If the $A1$ block does not have the transversal $S$ gate, we can realise the initialisation in the state $\ket{y+}$ and measurement in the basis $Y$ in the following way. First, we need another ancilla block $A2$, in which logical qubits are prepared in the distilled $\ket{y+}$ state. Second, we effectively initialise ancilla logical qubits in $A1$ in the state $\ket{y+}$ by applying measurements in the form $-Y_{A1}Y_{A2} = X_{A1}X_{A2}Z_{A1}Z_{A2}$, where $X_{A1},Y_{A_1},Z_{A_1}$ ($X_{A2},Y_{A_2},Z_{A_2}$) are $X,Y,Z$ operators of a logical qubit in the block $A1$ ($A2$). Because the measurement is of the $[\sigma_X,\sigma_Z] = 0$ type, we can realise the measurement through the block $A0$. Because the measurements are in the $Y$ basis, states of logical qubits in $A2$ are preserved, i.e. we do not need to re-prepare the distilled $\ket{y+}$ state. Finally, in the same way, we can measure ancilla logical qubits in $A1$ in the $Y$ basis. Though the above approach may increase the time cost compared with the transversal $S$ gate, the eventual time cost is still independent of the number of operators to be measured. 

\section{Costs}
\label{app:costs}

In devised sticking, we use a finely devised glue code to construct the measurement sticker. For the measurement sticker, we need a sufficiently large $d_R$ of the repetition code to maintain the code distance of the deformed code, i.e. $d_R =d$. Therefore, a finely devised measurement sticker requires $\Theta((n_G+r_G)d_R) = O(n_Ndq)$ physical qubits. Here, we have used $n_G,r_G = O(n_N+k_Nq)$ (see Theorem \ref{the:sticking}) and $k_N = O(n_N)$. 

If the simultaneous measurement has a logical thickness of $t$, the finely devised glue code has parameters $n_G,r_G = O(n_N+k_Nt) = O(n_Nt)$ (see Sec. \ref{app:thickness}). Then, a finely devised measurement sticker requires $\Theta((n_G+r_G)d_R) = O(n_Ndt)$ physical qubits. 

\begin{figure}[tbp]
\centering
\includegraphics[width=0.5\linewidth]{\figures/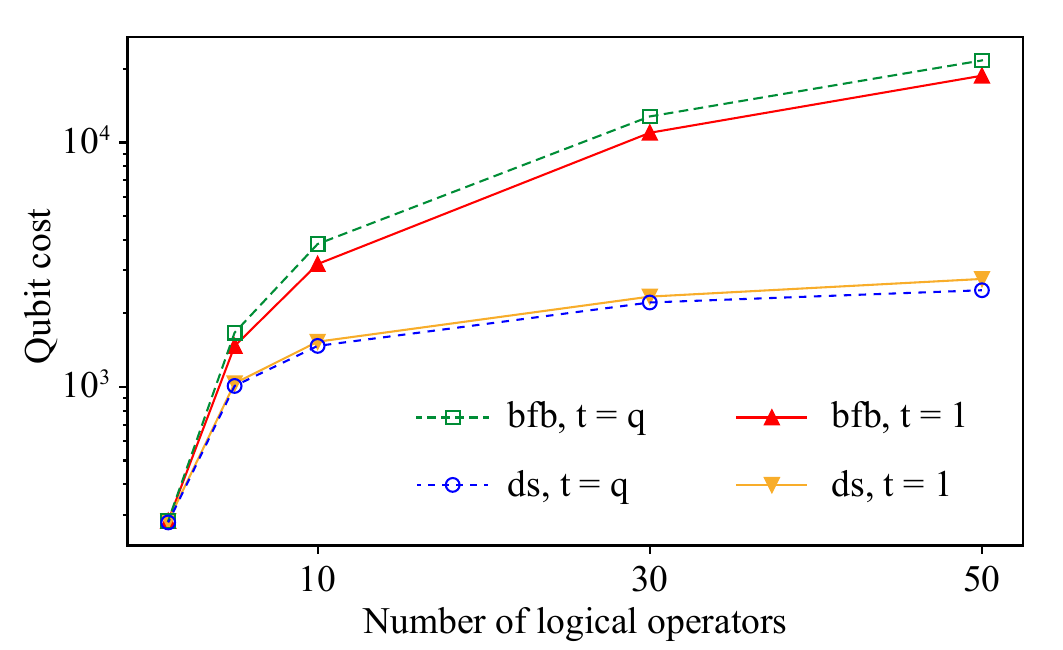}
\caption{
Median values of the qubit number required in simultaneous measurements on a [[578,162,5]] code \cite{Panteleev2021,Kovalev2012}. For each number of logical operators $q$, we randomly generate the operator set $\Sigma$ for one hundred times. The set $\Sigma$ consists of $Z$ logical operators acting non-trivially on up to $L = 5$ logical qubits, and its logical thickness is $t$. For each $\Sigma$, we evaluate the qubit costs in devised sticking (ds) and brute-force branching (bfb). 
}
\label{fig:costII}
\end{figure}

In brute-force branching, we use coarsely devised glue codes to construct branch stickers. For branch stickers, we take $d_R = 2$. Branch stickers on different levels are different in size. Suppose we want to measure logical operators acting non-trivially on up to $L$ logical qubits, and assume that each logical operator has a weight of $O(Ld)$. If a branch sticker acts on $p$ logical operators, its glue code has the parameter $n_N = O(pLd)$. Then, on the first level, there are two branch stickers, and each of them has $O(Ldq/2)$ physical qubits ($p = q/2$); on the second level, there are four branch stickers, and each of them has $O(Ldq/4)$ physical qubits ($p = q/4$); and so on. Therefore, we need $O(Ldq\log q)$ physical qubits to construct the branch stickers, and we need $O(Ld^2q)$ physical qubits to construct the measurement stickers used in brute-force branching; the total qubit cost is $O(Ldq(d+\log q))$. For each measurement sticker, we need $O(Ld^2)$ physical qubits ($n_N = O(Ld)$ and $q = 1$ for this measurement sticker), and we need $q$ measurement stickers. 

We can reduce the qubit cost by taking a small $q$. However, when $q$ is small, we only measure a small number of operators in each simultaneous measurement, meaning that we need more simultaneous measurements to complete a circuit, i.e. the time cost is increased. When $L$ and $q$ are constants with respect to code parameters $n$, $k$ and $d$, the qubit cost is $O(d^2)$; for hypergraph product codes, $O(d^2) = O(n)$, meaning that the qubit overhead is a constant. Accordingly, the time cost is $O(kd)$ when we measure $\Theta(k)$ logical operators. The overall spacetime cost is $O(k^{5/2})$ for hypergraph product codes, which is the same as the surface code. In the numerical results, we find that devised sticking has a smaller qubit cost than brute-force branching, suggesting that devised sticking has a smaller spacetime cost, i.e. its spacetime cost may be smaller than the surface code when the qubit cost is $O(n)$; however, future research is necessary to verify this conjecture. 

When $Ld$ is small compared with $n_N$, the brute-force branching outperforms devised sticking in the bound analysis. Notice that $n_N$ usually increases with $q$. However, in the numerical results, we find that devised sticking always has a smaller qubit cost, as shown in Fig. 5 in the main text and Fig. \ref{fig:costII}. 

For measuring general logical Pauli operators, we need ancilla logical qubits. For each operator to be measured, we need one (or three in some cases) ancilla logical qubit(s). By performing two rounds (or four rounds in some cases) of simultaneous measurements on $X$ and $Z$ logical operators that involve ancilla logical qubits, we can achieve the simultaneous measurement of an arbitrary set of logical Pauli operators. Therefore, the qubit (time) cost is amplified by a constant factor for measuring general logical Pauli operators. 

To implement controlled-NOT, Hadamard and $S$ gates using measurements, the time and qubit overhead is $O(1)$. Therefore, using the simultaneous measurement, we can implement a layer of the gates with the spacetime cost $O(nd)\times O(d)$ (we have taken $n_N = O(n)$ and $t = 1$). Here, the cost of $S$ gates does not include the preparation and distillation of $\ket{y+}$ magic states. Notice that the implementation of $S$ gates does not consume $\ket{y+}$ magic states, therefore, we only need to prepare and distill $\ket{y+}$ magic states at the beginning of the computing. Because we can apply measurements on any pair of logical qubits, we can directly realise the controlled-NOT gate on any pair of logical qubits. 

\section{Comparison to the protocol in Ref. \cite{Cohen2022}}
\label{app:comparison}

In Ref. \cite{Cohen2022}, an ancilla system is proposed for measuring a single logical Pauli operator, and a protocol for simultaneously measuring two commutative logical Pauli operators is also presented. While this approach can be generalised to multiple commutative operators by directly coupling one ancilla system to each operator, it encounters a problem due to overlapping logical operators, as discussed in the main text. Specifically, such overlap can lead to a deformed code that is no longer a quantum LDPC code. At the end of this section, we provide a rigorous analysis demonstrating that if the goal is to measure $\Theta(k)$ independent logical operators simultaneously, directly coupling one ancilla system to each operator inevitably violates the LDPC condition. 

In this work, we propose two protocols for measuring an arbitrary set of commutative logical Pauli operators simultaneously. In our protocols, the deformed code is always a quantum LDPC code. In the devised sticking protocol, we can measure an arbitrary set of $X$ or $Z$ logical operators simultaneously with one ancilla system. We also propose the brute-force branching protocol to separate an arbitrary set of $X$ or $Z$ logical operators such that we can measure all of them simultaneously with multiple ancilla systems (each of the ancilla systems could be the one proposed in Ref. \cite{Cohen2022}). Based on the measurement of $X$ or $Z$ logical operators, we can measure an arbitrary set of general logical Pauli operators simultaneously (up to the commutativity condition) by using ancilla logical qubits. 

The devised sticking protocol works because we find an algorithm to generate a proper glue code, called finely devised glue code, for an arbitrary operator set $\Sigma$. With the glue code, the ancilla system only measures operators in $\Sigma$ leaving other logical operators unmeasured. The key component of the brute-force branching protocol is the branch sticker introduced in this work. Through the theoretical analysis, we prove that a branch sticker acts on logical operators by transferring them instead of measuring them, and a branch sticker only requires a small $d_R$ to maintain the code distance. The small $d_R$ is important for reducing the qubit cost in brute-force branching. 

{\bf Technical comparison.} For a detailed comparison, we review the protocol in Ref. \cite{Cohen2022} in our framework. The ancilla systems proposed in Ref. \cite{Cohen2022} for measuring $X$ and $Z$ logical operators are instances of measurement stickers. To measure a $Z$ logical operator, denoted by $\bar{Z}$, it is assumed that there are no other $Z$ logical operators on the support $\mathcal{Q}(\bar{Z})$. Under the assumption, taking the naked glue code is sufficient for measuring $\bar{Z}$, which is exactly the protocol illustrated in Fig. 2 in Ref. \cite{Cohen2022} (but for a $Z$ operator). To measure the product of two $Z$ logical operators $\bar{Z}_1\bar{Z}_2$, it is assumed that $\mathcal{Q}(\bar{Z}_1)$ and $\mathcal{Q}(\bar{Z}_2)$ do not overlap, and there is no other $Z$ logical operators on the support $\mathcal{Q}(\bar{Z}_1)\cup\mathcal{Q}(\bar{Z}_2)$. In this case, the naked glue code is insufficient because both $v_1$ and $v_2$ [$\bar{Z}_1 = Z(v_1)$ and $\bar{Z}_2 = Z(v_2)$] are in $(\mathrm{ker}H_N)S_N$. We need to remove $v_1$ and $v_2$ and leave $v_1+v_2$ in the space. This can be achieved by taking the dressing matrix $D = \begin{pmatrix} 1 & 1 & 0 & 0 & \cdots \end{pmatrix}$, where the first entry is on the support $\mathcal{Q}(\bar{Z}_1)$, and the second entry is on the support $\mathcal{Q}(\bar{Z}_2)$. By taking such a dressed glue code, we obtain the protocol illustrated in Fig. 3 in Ref. \cite{Cohen2022} (but for a $Z$ operator). In these examples, the protocol in Ref. \cite{Cohen2022} corresponds to naked glue codes and instances of dressed glue codes, which work under certain assumptions. 

In comparison, we propose the general formalism of measurement and branch stickers. We find that all stickers with compatible glue codes (provided with proper pasting matrices) can be used for certain logical operations. We present the criteria for choosing glue codes, called coarsely devised glue codes and finely devised glue codes. In one of our protocols called devised sticking, we use measurement stickers with finely devised glue codes to achieve simultaneous measurements of logical Pauli operators. We give algorithms for generating finely devised glue codes. The generated codes are finely devised LDPC glue codes, in contrast to naked glue codes and instances of dressed glue codes taken in Ref. \cite{Cohen2022}. By using finely devised LDPC glue codes, we can measure arbitrary logical Pauli operators simultaneously while maintaining the low density of parity checks, without any assumption about supports. 

Besides devised sticking, we also propose brute-force branching, which is another protocol that can measure arbitrary logical Pauli operators simultaneously while maintaining the low density of parity checks. In brute-force branching, we use branch stickers in addition to measurement stickers. 

Regarding rigorous theoretical results, the protocol in Ref. \cite{Cohen2022} is based on three lemmas and one theorem. In our framework, these results focus on measurement stickers with naked glue codes under the assumption of supports. The measurement sticker with a naked glue code $H_N$ measures all logical operators on $\mathcal{B}_N$, which may include logical operators that need not be measured. Our protocol is based on two theorems and two lemmas. Theorem \ref{the:GLS} applies to general measurement stickers with finely devised glue codes and branch stickers with coarsely devised glue codes, without any assumption of supports. Theorem \ref{the:sticking} states the existence of coarsely and finely devised glue codes satisfying the LDPC condition; and its proof contains the algorithms for generating proper glue codes. Upper bounds for cost factors due to glue codes are also given in Theorem \ref{the:sticking}. Lemmas \ref{lem:general} and \ref{lem:regularisation} justify the simultaneous measurement of general logical Pauli operators. 

{\bf Parallelisation through optimising logical operators.} As a potential way of overcoming the overlap between logical operators (the problem of a large crowd number), we can optimise the representatives of logical operators. For example, let $\bar{Z}_1,\bar{Z}_2\in \mathcal{Z}$ be two logical operators that have overlap on some physical qubits. We may be able to find a stabiliser operator $g\in \mathcal{S}$ such that $\bar{Z}_1$ and $g\bar{Z}_2$ do not overlap. Then, we can measure them with two ancilla systems. This approach never works for full parallelisation, although it may work to a certain extent. Think of that we want to measure $q = \Theta(k)$ independent logical operators. Each operator has a weight of $\Omega(d)$. Then, $\Omega(kd)$ is the total weight of logical operators in $\Sigma$, i.e. $\sum_{\sigma\in\Sigma}\vert\mathcal{Q}(\sigma)\vert = \Omega(kd)$. The physical qubit number is $n$. Therefore, the average crowd number is $\Omega(kd/n)$. 

The average is a lower bound of the maximum crowd number. If the maximum crowd number is smaller than $\Omega(kd/n)$ for all physical qubits, the average is smaller than $\Omega(kd/n)$, leading to a contradiction. For a quantum LDPC code with a good encoding rate, i.e. $k = \Theta(n)$, the maximum crowd number is $\Omega(d)$. Notice that this lower bound of the maximum crowd number holds for arbitrary representatives of logical operators. Therefore, the maximum crowd number always increases with the code distance regardless of the optimisation of logical operators. This means that we have to couple at least $\Omega(d)$ ancilla systems to a physical qubit, which breaks the LDPC condition. 

\begin{figure}[tbp]
\centering
\includegraphics[width=\linewidth]{\figures/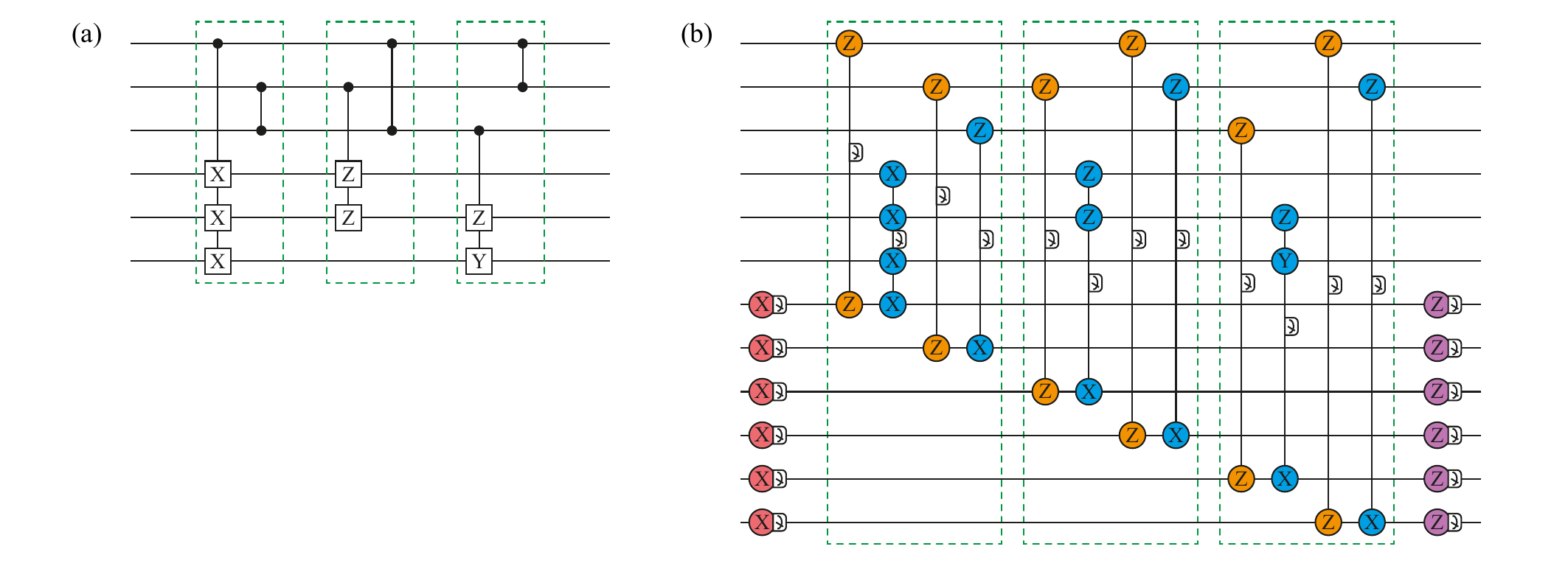}
\caption{
(a) A circuit with six commutative gates. (b) The measurement circuit for realising the gate circuit in (a). 
}
\label{fig:parallelism}
\end{figure}

\section{Comparison between conventional parallelism and ultimate parallelism}
\label{app:parallelism}

As an example, we illustrate a circuit with six gates in Fig. \ref{fig:parallelism}(a). These six gates commute with each other. In each green box, the two gates act on different qubits. In conventional parallelism, we can simultaneously implement the two gates in the same green box; then, the circuit requires three time steps. In ultimate parallelism, we can simultaneously implement all six gates; then, the circuit requirements only one time step. 

The circuit can be realised through logical measurements as shown in Fig. \ref{fig:parallelism}(b). Commutative measurements are marked with the same colour (red, orange, blue and violet). In each green box, the two commutative measurements act on different qubits. In conventional parallelism, we can simultaneously implement the two commutative measurements in the same green box; then, the circuit requires eight time steps. In ultimate parallelism, we can simultaneously implement all commutative measurements; then, the circuit requires only four time steps. Notice that even if we want to implement one gate, we need four time steps to realise using measurements. 

\end{widetext}

\bibliography{references.bib}

\end{document}